\documentclass[preprint,12pt,compress]{elsarticle}

\usepackage{amssymb}

\usepackage{amsthm}

\usepackage{amsfonts}
\usepackage{graphicx}
\usepackage{fancyhdr}
\usepackage{epstopdf}
\usepackage{algorithmic}

\ifpdf
  \DeclareGraphicsExtensions{.eps,.pdf,.png,.jpg}
\else
  \DeclareGraphicsExtensions{.eps}
\fi

\usepackage{amssymb}
\usepackage{mathtools}

\usepackage[mathscr]{eucal}
\usepackage{color}

\usepackage{tikz}
\usetikzlibrary{shapes,arrows}
\usetikzlibrary{decorations.pathreplacing,angles,quotes}

\usepackage{todonotes}
\usepackage{dsfont}
\usepackage{cleveref}
\usepackage[normalem]{ulem}

\makeatletter
\def\ps@pprintTitle{%
 \let\@oddhead\@empty
 \let\@evenhead\@empty
 \def\@oddfoot{\centerline{\thepage}}%
 \let\@evenfoot\@oddfoot}
\makeatother

\newcommand{\mb}{\mathbb}

\newcommand{\mc}{\mathcal}
\newcommand{\br}{\breve}
\newcommand{\eps}{\epsilon}
\newcommand{\mbR}{{\mathds{R}}}
\newcommand{\Id}{{\mathbb{I}}}

\newcommand{\mcD}[1]{{\mathcal{D}#1}}

\newcommand{\mfN}{{\mathfrak{N}}}
\newcommand{\mfK}{{\mathfrak{K} }}
\newcommand{\mfT}{{\mathfrak{T} }}
\newcommand{\hQ}{{\widehat{Q}}}
\newcommand{\hH}{{\widehat{H}}}

\newtheorem{theorem}{Theorem}

\newtheorem{definition}{Definition}

\newtheorem{remark}{Remark}
\newtheorem{assumption}{Assumption}

\tikzstyle{block1} = [rectangle, draw, thick,fill=blue!10, text width=4.5em, text centered, rounded corners, minimum height=2em]

\tikzstyle{block2} = [rectangle, draw, thick,fill=blue!10, text width=2em, text centered, rounded corners, minimum height=2em]

\tikzstyle{line} = [draw, -latex']

\begin{document}

\begin{frontmatter}




\title{{\normalsize\hspace{0cm}\vspace{-3cm}**To Appear in \textit{Systems \& Control Letters}**}\vspace{4cm}\\Convex Analysis for LQG Systems with Applications to Major Minor LQG Mean-Field Game Systems \footnote{DF would like to acknowledge the support of Fonds de Recherche du Qu\'ebec -- Nature et Technologies (FRQNT). SJ would like to acknowledge the support of the Natural Sciences and Engineering Research Council of Canada
(NSERC), [funding reference numbers RGPIN-2018-05705 and RGPAS-2018-522715]. PEC would like to acknowledge the support of NSERC.}}

\author[dena]{Dena Firoozi}
\author[seb]{Sebastian Jaimungal}
\author[peter]{Peter E. Caines}
\address[dena]{Department of Decision Sciences, HEC Montr\'{e}al, Montreal, QC, Canada\\ (email: dena.firoozi@hec.ca)}
\address[seb]{Department of Statistical Sciences, University of Toronto, Toronto, ON, Canada, (email: sebastian.jaimungal@utoronto.ca)}
\address[peter]{Centre for Intelligent Machines (CIM) and the Department of Electrical and Computer Engineering (ECE), McGill University, Montreal, QC, Canada, \\(email: {peterc@cim.mcgill.ca})}

\begin{abstract}
We develop a convex analysis approach for solving LQG optimal control problems and apply it to major-minor (MM) LQG mean-field game (MFG) systems. The approach retrieves the best response strategies for the major agent and all minor agents that attain an $\epsilon$-Nash equilibrium. An important and distinctive advantage to this approach is that unlike the classical approach in the literature, we are able to avoid imposing assumptions on the evolution of the mean-field. In particular, this provides a tool for dealing with complex and non-standard systems. 
\end{abstract}



\begin{keyword}
  convex analysis \sep LQG systems \sep major-minor mean-field games
\end{keyword}
 \end{frontmatter}

\section{Introduction}
Various approaches such as calculus of variations, stochastic maximum principle, dynamic programming, square completion and change of functional have been used to address deterministic linear quadratic (LQ) and stochastic linear quadratic (LQG) optimal control problems \cite{Kirk1971, XYZBook1999, bensoussan1992, Anderson1990}.
The convex analysis approach to optimization for static systems, instead, uses the criteria of a vanishing G\^ateaux derivative of  the optimization functional (see e.g., \cite{ConvecAnalysisBook1999}, \cite{Allaire2007}) to seek optimality. \cite{CarmonaBSDEsBook2016} establishes the relationship between the G\^{a}teaux derivative of the cost functional of a stochastic dynamical system and its Hamiltonian. The convex analysis approach to stochastic optimization has been applied in several financial related context, including, \cite{cvitanic1992convex} which studies portfolio optimization, \cite{bank2017hedging} which studies a stochastic tracking problem in finance, \cite{JaimungalPhil2018, casgrain2018mean} which analyses an algorithmic trading problem with heterogeneous agents, and \cite{ShrivatsFirooziJaimungal2020} formulates solar renewable energy certificate (SREC) markets using mean-field game (MFG) methodology. 


 MFG theory establishes the existence of approximate Nash equilibria and the corresponding individual strategies for stochastic dynamical systems in games involving a large number of agents. The equilibria are termed $\epsilon$-Nash equilibria and are generated by the local, limited information control actions of each agent in the population. The control actions constitute the best response of each agent with respect to the behaviour of the mass of agents. Moreover, the approximation error, induced by using the MFG solution, converges to zero as the population size tends to infinity. The analysis of this set of problems originated in \cite{HuangCDC2003a, HuangCDC2003b, HuangCIS2006, HuangTAC2007} (see \cite{CHM17}), and independently in \cite{Lasry2006a, Lasry2006b, Lasry2007}. Many extensions and generalizations of mean-field games exist, principally among them are the probabilistic formulation \cite{CarmonaDelarue2013, CarmonaDelarueBook2018}, the master equation approach \cite{CardaliaguetMasterEqBook2019} and mean-field type control theory \cite{BensoussanBook2013}.  
 
 In \cite{Huang2010} and \cite{Huang2012} the authors, for the first time, analyse and solve mean-field game systems where there is a major agent (whose impact does not vanish in the limit of infinite population size)  together with a population of minor  agents (where each agent  has individually asymptotically negligible effect). Major-minor problems lead to stochastic mean fields; by extending the state space for such systems, \cite{Huang2010} and \cite{NourianSiam2013} establish the existence of closed-loop $\epsilon$-Nash equilibria together with the individual agents' control laws that yield the equilibria, respectively, for an LQG case and a general nonlinear case. Following this work, the partially observed MFG theory for nonlinear and linear quadratic systems with major and minor agents is developed in \cite{Kizilkale2013, Kizilkale2014, KizilkaleTAC2016, SenCDC2014, SenACC2015, SenSiam2016, CainesSen2019, FirooziCDC2015, FirooziCainesTAC} where the agents' states are partially observed in additive Brownian noise. 

Other works that seek an $\epsilon$-Nash equilibrium between a major agent and a large number of minor agents in a mean-field game setup include \cite{huang2015mean, Kordonis2015, CarmonaZhu2016, CarmonaWang2017, FJC-arXiv2018, FPC-arXiv2018, HuangMA2020, FirooziCDC2017,FirooziCDC2016,FirooziIFAC2017}. In \cite{CarmonaZhu2016}, using the probabilistic approach to mean-field games, the authors establish the existence of an open-loop $\epsilon$-Nash equilibrium for a general case and provide explicit solutions for an LQG case.  In \cite{CarmonaWang2017} an alternative formulation for mean-field games with major and minor agents is proposed,  where the search for Nash equilibria in the infinite-population limit is framed as the search for fixed points in the space of controls for the best response function for the major and minor agents. This permits the authors to formulate the open-loop and closed-loop versions of the problem simultaneously. Subsequently they retrieve the open-loop solutions of \cite{CarmonaZhu2016}.

The works \cite{LasryLions2018, CardaliaguetCirant2018} characterize the Nash equilibrium for a general mean-field game system with one major agent and an infinite number of minor agents via the MFG  Master Equations.  In \cite{CardaliaguetCirant2018} it is shown that for mean field games with a finite number $N$ of minor agents and a major agent, the solution of the corresponding Nash system converges to the solution of the system of master equations as $N$ tends to infinity. It is to be noted that for the LQG case it has been demonstrated in \cite{HuangCIS2020} that the LQG MM MFG Master Equations yield the original LQG MM MFG equations of  \cite{Huang2010}. 

Another line of research in this area aims to characterize a Stackelberg equilibrium between the major agent and the population of minor agents, while the minor agents reach a Nash equilibrium among themselves in the infinite population limit \cite{BensoussanSICON2015, BasarMoon2016, Bensoussan2016, BensoussanSICON2017, BasarMoon2018, FuHorst2018, JiangZhang2019, EliePossamai2019}.

The recent work \cite{FirooziCainesCDC2019} studies LQG mean-field games with two major agents and a large population of minor agents and finds the explicit closed-loop $\epsilon$-Nash strategies for all agents. 

In this work, we take a convex analysis approach to derive the solutions to LQG optimal control problems. We then apply the methodology to major minor (MM) LQG MFG systems addressed in \cite{Huang2010} to retrieve the best response strategies for the major agent and each individual minor agent (for examples and applications of such systems see \cite[Sec 7.1]{Huang2010}, \cite[Sec 4]{FirooziCainesTAC}, \cite{FirooziISDG2017, FirooziPakniyatCainesCDC2017}, and for a formulation of a class of mean-field games as equivalent optimization problems see \cite{LiZhangMTNS2018}). The main features of the convex analysis approach presented here are summarized in the following points:
\begin{itemize}
\item	In the classical approach towards MM LQG MFG (see e.g. \cite{Huang2010, Huang2012, KizilkaleTAC2016, FirooziCDC2015, FirooziCainesTAC,FJC-arXiv2018,FPC-arXiv2018}), one first derives heuristic dynamics  for the mean-field in the infinite population limit, this is done by choosing the structure of the control action for a generic minor agent as a function of its own state and of the system's mean field, where the control coefficients are unknown and distinct for each subpopulation of agents. Next the control actions are substituted in the minor agents' dynamics, and it is  shown that the average of minor agents closed-loop dynamics, as the number of agents goes to infinity, converges to the so-called mean field equation. Finally one extends the major and the minor agents' state space with the mean field to establish the existence of an $\epsilon$-Nash equilibrium and to obtain the individual agents' control laws that yield the equilibrium. In the end, the mean field equation parameters are obtained through the  consistency equations, these effectively equate the nominal  mean field employed in the initial form of the individual agents' control laws with the mean field resulting from the collective action of the mass of agents using these laws. 

However, it is not always easy, a priori, to choose a structure for a generic minor agent's control action which includes the correct processes and parameters; especially for complex and non-standard problems this would be a difficult job. For example, for the case where minor agents have partial observations, the mean field will be a function of the average of estimation errors of agents, and for the case where major agent has also partial observations, the mean field equation would also be a function of the major agent's estimates of its own state and the mean field (see \cite{FirooziCainesTAC}). Another example is the case where (unobserved) latent factors appear in the dynamics of the agents (see \cite{casgrain2018mean,FJC-arXiv2018,FJPCDC2018}). In contrast, our approach makes no prior assumptions on the structure of minor agents' control actions and the evolution of the mean-field. Instead we derive the mean-field equation while solving the infinite-population limit of the problem. This alleviates the a priori assumptions on the structure of control actions. 


\item In the classical approach, it may be thought that it is not clear whether the Nash equilibrium and the corresponding $\epsilon$-Nash equilibrium are consequences of the assumptions imposed a priori. The approach we take here verifies that one obtains the same Nash equilibrium (and hence $\epsilon$-Nash equilibrium) as in the classical approach without imposing prior assumptions and hence answers the question about the resulting equilibria. 

\item	The approach taken here helps to elucidate the impact that agents, both major and minor, have on one another and on the overall system. Specifically, the analysis shows how a perturbation in an agent's control action propagates through the system and how it affects other agents. Hence, it provides a tool for understanding the complicated interactions of agents in large-scale systems. 

\item The method is applicable to non-standard single/multi agent stochastic problems (e.g. \cite{casgrain2018mean,FJC-arXiv2018,FJPCDC2018}), where the classical results do not apply, and extending the classical approaches such as maximum principle and dynamic programming may be difficult. This is because the formulation in the manuscript readily generalizes to these more sophisticated cases. Moreover, it leads to a direct characterization of the optimal control actions in stochastic optimal control problems.



\end{itemize}

The remainder of this paper is organized as follows. \Cref{sec:overview} provides a convex analysis approach using G\^ateaux derivatives for static systems. \Cref{sec:LQGproblems} extends these methods to construct the solution to single-agent LQG problems. Finally, in \Cref{sec:MMLQGMFG}, we further extend the approach so as to retrieve the best response strategies for MM LQG MFG systems.

\section{Convex Analysis Overview}\label{sec:overview}

Let $V$ be a reflexive Banach space, with corresponding dual space $V^*$, and $\mathscr{V}$ be a non-empty closed convex subset of $V$.
\begin{definition}[G\^ateaux Derivative \cite{ConvecAnalysisBook1999, Allaire2007}]\label{def:GateauxDerivative} The function $J$ defined on a neighbourhood of $u \in V$ with values in $\mbR$ is G\^ateaux differentiable at $u$ in the direction of $\omega\in V$ if there exists $\mcD{J(u)} \in V^*$ such that
\begin{align} \label{GateauxDerivative}
\langle  \mcD{J(u)}, \omega \rangle = \lim_{\epsilon \rightarrow 0} \frac{J(u+\epsilon \,\omega)-J(u)}{\epsilon}.
\end{align}
The function $\mcD{J(u)}$ is called the G\^ateaux derivative of $J$ at $u$.
\end{definition}
\begin{theorem}[Euler Inequality] \label{thm:EulerIneq} Assume  $J$ is convex, continuous, proper, and G\^ateaux differentiable with continuous derivative $\mcD{J(u)}$. Then
\begin{equation}
J(u) = \inf_{v \in \mathscr{V}} J(v),
\end{equation}
if and only if
\begin{equation} \label{EulerIneq}
\langle \mcD{J(u)}, v-u \rangle \geq 0,   \quad \forall v \in \mathscr{V}.
\end{equation}
\end{theorem}
\hfill $\square$

\begin{remark}[Euler Equality] \label{coro:convexAnalysis} In \Cref{thm:EulerIneq} for the case where $\mathscr{V}=V$, $\omega=v-u$ generates the whole space $V$, and therefore \eqref{EulerIneq} reduces to Euler equality
\begin{equation} \label{EulerEq}
\langle \mcD{J(u)}, \omega \rangle = 0,   \quad \forall \omega \in V,
\end{equation}
which implies that
\begin{equation}
J(u) = \inf_{v \in \mathscr{V}} J(v) \quad \Leftrightarrow \quad \mcD{J(u)} =0.
\end{equation}
\end{remark}

The Banach space in this paper is the space of square-integrable $\mbR^m$-valued measurable functions, which will be specified in detail for single-agent LQG systems and major minor LQG mean-field game systems in \Cref{sec:LQGproblems} and \Cref{sec:MMLQGMFG}, respectively.

\section{Single-Agent LQG Problems}
\label{sec:LQGproblems}

In this section, we rederive the solution to single-agent LQG problems  using a convex analysis method.
\subsection{Dynamics}
Consider single-agent LQG systems with governing dynamics
\begin{align}\label{SLQConvsys}
dx_t = (A\, x_t + B\,u_t +b(t) )\,dt + \sigma(t)\, dw_t,
\end{align}
where $t \in \mfT:=[0,T]$, and the continuous processes $x_t \in \mbR^n$, $u_t \in \mbR^m$, and $w_t \in \mbR^r$  denote, respectively, the state, the control action, and a standard Wiener process. Moreover, $A \in \mbR^{n \times n}$, $B \in \mbR^{n \times m}$ are constant, and $b(t) \in \mbR^n$, $\sigma(t)\in \mbR^{n \times r}$, are deterministic, continuous and bounded functions of time.

\subsubsection{Control $\sigma$-Fields}
We denote by $\mc{F}\coloneqq (\mc{F}_{t})_{t\in\mfT}$ the natural filtration generated by the agent's state $(x_{t})_{t\in\mfT}$. The admissible set of controls $\mc{U}$ is the set of continuous linear state feedback control laws $u_t = u(t,x_t), {t\in\mfT},$ that are $\mc{F}$-adapted $\mbR^m$-valued processes such that $\mb{E}[\int_0^T u_t^\intercal u_t dt] < \infty$, for any finite $T$.

\subsection{Cost Functional}
The cost functional to be minimized is given by 
\begin{equation}\label{SLQconvCost}
J(u) =   \tfrac{1}{2}  \mathbb{E} \Big[ e^{-\rho T}x_T^\intercal \hQ x_T +\int_{0}^{T} e^{-\rho t} \Big \{ x_t^\intercal Qx_t + 2\,x_t^\intercal N u_t + u_t^\intercal R u_t -2\, x_t^\intercal \eta - 2\, u_t^\intercal n \Big \}\,dt  \Big],
\end{equation}
where $\hQ , Q \in \mbR^{n \times n}$, $N \in \mbR^{n \times m}$, $R \in \mbR^{m \times m}$, $\eta \in \mbR^n$, $n \in \mbR^m$, denote the weight matrices and  $\rho \in \mbR$ denotes the discount rate.
\begin{assumption}\label{ass:singleAgentConvCond}
For the cost functional \eqref{SLQconvCost} to be convex, we assume $\hQ  \geq 0$, $R>0$, and $Q - N R^{-1} N^\intercal \geq 0$.
\end{assumption}

\subsection{Optimal Control Action}

The system dynamics \eqref{SLQConvsys} together with the cost functional \eqref{SLQconvCost} constitute a stochastic LQ optimal control problem, which we  solve  using the following theorem.
\begin{theorem}[G\^ateaux Derivative of Cost for LQG Systems] \label{thm:LQGGatDeriv}For the class of LQG systems described by \eqref{SLQConvsys}-\eqref{SLQconvCost}, the G\^ateaux derivative of the cost functional is
\begin{multline}\label{LQGGatDerivFinal}
\langle \mcD{J(u)}, \omega \rangle = \mathbb{E} \bigg [\int_0^T \omega_t^\intercal \bigg\{e^{-\rho t}N^\intercal x_t + e^{-\rho t} R\,u_t -e^{-\rho t} n \\
+ B^\intercal \Big(e^{-A^\intercal t} M_t -\int_0^t e^{-\rho s} e^{A^\intercal(s-t)} (Q x_s + N u_s - \eta )ds \Big )  \bigg  \} dt \bigg ],
\end{multline}
where $\omega$ lies in the space $\mc{U}$ and $(M_t)_{t\in\mfT}$ is a martingale given by
\begin{align}
M_t =  \mathbb{E} \Big[e^{-\rho T} e^{A^\intercal T} \hQ  x_ T +\int_0^T e^{-\rho s}e^{A^\intercal s }  (Q x_s + N u_s- \eta \big )ds \Big | \mc{F}_t \Big].
\label{eqn:DefM}
 \end{align}
\end{theorem}
\hfill $\square$
\begin{proof}
The solution $x_t$ to the state representation of the system \eqref{SLQConvsys} subject to the control action $u_t$ is
\begin{equation}\label{stateOrig}
x_t = e^{At} x_0 + \int_0^t e^{A(t-s)}\big(B u_s +b(s)\big) ds + \int_0^t e^{A(t-s)}\sigma(s) dw_s,
\end{equation}
where $x_0 \in \mbR^n$ and $\phi(t,s)=e^{A(t-s)},\,\forall\, s \leq t \leq T,$ denote, respectively, the initial state and the state transition matrix for the system $\eqref{SLQConvsys}$.

Let $x^{\epsilon}$ denote the solution to \eqref{SLQConvsys} subject to a perturbed control action $u+\epsilon \omega$ in the direction of $\omega \in \mc{U}$. It is given by
\begin{equation}\label{statePert}
x_t^{\epsilon} = e^{At} x_0 + \int_0^t e^{A(t-s)}\big(B u_s +b(s)\big)\, ds + \int_0^t e^{A(t-s)}\sigma(s)\, dw_s + \epsilon \int_0^t e^{A(t-s)} B \omega_s\, ds.
\end{equation}
Substituting \eqref{stateOrig} into \eqref{statePert} implies
\begin{equation}\label{sOrigVSsPert}
x_t^{\epsilon} = x_t + \epsilon \int_0^t e^{A(t-s)} B \omega_s \,ds,
\end{equation}
or equivalently (from  differentiating both sides)
\begin{align} \label{sPertEvol}
d{x}^{\epsilon}_t = d{x}_t + \epsilon B \omega_t\, dt + \epsilon A \int_0^t e^{A(t-s)}B\omega_s\,ds.
\end{align}
The cost of the perturbed control action $u + \epsilon \omega$ and the corresponding perturbed state $x^{\epsilon}$ is
\begin{multline}
J(u+\epsilon \omega) = \tfrac{1}{2}\mb{E} \bigg [e^{- \rho T}(x_T^{\epsilon})^\intercal \hQ  x_T^{\epsilon} + \int_0^T e^{- \rho s}\Big\{ (x_s^{\epsilon})^\intercal Q x_s^{\epsilon}
+ 2(x_s^{\epsilon})^\intercal N (u_s+\epsilon \omega_s)\\ + (u_s+\epsilon \omega_s)^\intercal R (u_s+\epsilon \omega_s) -2 (x_s^{\epsilon})^\intercal
\eta - 2(u_s+\epsilon \omega_s)^\intercal n  \Big \} ds \bigg].\label{CostPert}
\end{multline}
Using integration by parts for It\^o processes \cite{Shreve1998}, allows us to rewrite the terminal cost as
\begin{multline}
e^{-\rho T}(x_T^{\epsilon})^{\intercal} \hQ  x_T^{\epsilon} = (x_0)^{\intercal} \hQ  x_0  + \int_{0}^T d\big(e^{-\rho s}(x_s^{\epsilon})^\intercal \hQ  x_s^{\epsilon}\big)\allowdisplaybreaks\\
=(x_0)^\intercal \hQ  x_0 -\rho \int_0^T e^{- \rho s} (x_s^{\epsilon})^\intercal \hQ  x_s^{\epsilon}\,ds+ 2 \int_{0}^T e^{-\rho s}(x_s^{\epsilon})^\intercal \hQ  \,dx_s^{\epsilon} \\
 + \int_0^T e^{-\rho s} \sigma(s)^\intercal \hQ \sigma(s)\,ds.
\label{TerminalCost}
\end{multline}

Substituting \eqref{TerminalCost} and then \eqref{sOrigVSsPert}-\eqref{sPertEvol} into \eqref{CostPert}, and collecting terms, we have
\begin{multline}  \label{CostPertVsCostOrig}
J(u+\epsilon \omega) = J(u) + \mb{E}\bigg[\epsilon \int_{0}^{T} e^{-\rho s} \bigg\{ \bigg(\int_0^s e^{A(s-t)}B \omega_t \,dt \bigg)^\intercal\bigg(\hQ  dx_s \allowdisplaybreaks
\\
+ (Q x_s+ N u_s+A^\intercal \hQ  x_s-\rho \hQ  x_s - \eta)\,ds\bigg) +
\big( (x_s)^\intercal N \omega_s + (x_s)^T \hQ  B \omega_s  \allowdisplaybreaks
\\
+ (u_s)^\intercal  R\omega_s -n^\intercal \omega_s \big) ds\bigg \}
+ \epsilon^2 \int_0^T e^{- \rho s}\bigg\{\big(\int_0^s e^{A(s-t)}B \omega_t dt\big)^\intercal\big(\hQ  A \int_0^s e^{A(s-t)}B \omega_t dt \allowdisplaybreaks\\-\rho \hQ  \int_0^s e^{A(s-t)}B \omega_t dt +\hQ  B \omega_s + Q \int_0^s e^{A(s-t)}B\omega_tdt + N \omega_s\big) + (\omega_s)^\intercal  R\omega_s\big)  \bigg\}ds\bigg].
\end{multline}
The G\^ateaux derivative of $\mcD{J(u)}$ in the direction of $\omega$ is obtained by subtracting $J(u)$, dividing both sides of the equation by $\epsilon$, and finally taking the limit as $\epsilon \rightarrow 0$. The result is
\begin{multline}\label{GatDeriv}
\langle \mcD{J(u)}, \omega \rangle = \mb{E} \bigg[\int_{0}^{T} e^{-\rho s}\bigg\{ \bigg(\int_0^s e^{A(s-t)}B \omega_t dt \bigg)^\intercal
\Big(\hQ  dx_s + (Q x_s+ N u_s+A^\intercal \hQ  x_s
\\
- \rho \hQ  x_s-\eta) ds\Big) +
\big( x_s^\intercal N \omega_s + x_s^\intercal \hQ  B \omega_s  + u_s^\intercal R\omega_s -n^\intercal \omega_s\big) ds\bigg \}\bigg].
\end{multline}
We substitute \eqref{SLQConvsys} to rewrite the double integral in \eqref{GatDeriv} as
\begin{multline}\label{intDecomp}
\int_{0}^{T} \bigg\{\int_0^s e^{-\rho s}\omega_t^\intercal B^\intercal e^{A^\intercal(s-t)}\hQ  \sigma(s)dtdw_s+
\int_0^s e^{-\rho s}\omega_t^\intercal B^\intercal e^{A^\intercal(s-t)}\\\times \Big(\hQ A x_s + \hQ Bu_s +\hQ b(s)+Q x_s+ N u_s+A^\intercal \hQ  x_s - \rho \hQ  x_s-\eta\Big) dtds\bigg\}.
\end{multline}
Let us denote by $\phi:= (\phi_{t,s})_{0\le t\le s \le T}$, where 
\begin{equation}
\phi_{t,s} = e^{-\rho s}\omega_t^\intercal B^\intercal e^{A^\intercal(s-t)}\hQ  \sigma(s). 
\end{equation}
Given that (i) $\phi_{t,s}$ is $\mc{F}_t\times \mc{B}(\mbR)$-measurable, where $\mc{B}(\mbR)$ denotes the smallest $\sigma$-algebra that contains the open intervals of $\mbR$, (ii) every realization of $\phi_{t,s}$ is bounded $\forall s\in \mbR,\, t \in \mfT$, (iii) $\int_t^T\phi_{t,s}dw_s$ is  $\mc{F}_t \times \mc{B}(\mbR)$-measurable $\forall t \in \mfT$; the conditions of the stochastic version of Fubini's theorem \cite{Ikeda2014} hold. Moreover, since the processes in the second term of \eqref{intDecomp} are continuous, the conditions of the ordinary Fubini's theorem hold \cite{StrangBook}. Applying Fubini's theorem to change the order of integration in \eqref{GatDeriv} results in
\begin{equation}\label{GatDerivFubini}
\begin{split}
&\langle \mcD{J(u)}, \omega \rangle
 = \mathbb{E} \bigg[ \int_0^T\!\! \omega_t^\intercal \bigg\{
e^{-\rho t} B^\intercal \hQ  x_ t+ e^{-\rho t} N^\intercal x_t + e^{-\rho t}R u_t - e^{-\rho t}n
\\
&+ B^\intercal \int_t^T \!\!\!e^{-\rho s}e^{A^\intercal(s-t)} \Big (\hQ  dx_s + \big(\left( (A^\intercal -\rho\Id_n) \hQ  + Q\right) x_s + N u_s - \eta \big )ds \Big) \bigg  \} \bigg ] dt,
\end{split}
\end{equation}
where we denote by $\Id_n \in \mbR^{n \times n}$ the identity matrix. Integrating by parts once again, we have
 \begin{equation}
 \begin{split}
& \int_t^T e^{A^\intercal(s-t)-\rho s}\left((A^\intercal-\rho\Id_n)\hQ x_s ds +\hQ dx_s\right) 
 \\
 &\qquad\qquad =
 \int_t^T d( e^{A^\intercal(s-t)}e^{-\rho s}\hQ x_s) 
 =   e^{-\rho T}e^{A^\intercal(T-t)}\hQ x_T - e^{-\rho t} \hQ  x_t,
 \end{split}
  \end{equation}
and substituting into \eqref{GatDerivFubini} yields
 \begin{multline} \label{GatDerivFubiniFinal}
\langle \mcD{J(u)}, \omega \rangle  = \mathbb{E} \bigg [\int_0^T \omega_t^\intercal \bigg\{ e^{-\rho T} B^\intercal e^{A^\intercal(T-t)} \hQ  x_T+ e^{-\rho t} N^\intercal x_t + e^{-\rho t}R u_t -e^{-\rho t} n \\
+ B^\intercal \int_t^T  e^{-\rho s}e^{A^\intercal(s-t) } (Q x_s + N u_s - \eta \big)ds \bigg \} dt\bigg] .
 \end{multline}
Using the smoothing property of conditional expectations \cite{CainesBook1988}, the G\^ateaux derivative \eqref{GatDerivFubiniFinal} may be rewritten as
 \begin{multline}\label{SLQGderiv2}
\langle \mcD{J(u)}, \omega \rangle = \mathbb{E} \bigg [\int_0^T \omega_t^\intercal \bigg\{ e^{-\rho t} N^\intercal x_t + e^{-\rho t} R u_t - e^{-\rho t} n \\+ B^\intercal \mathbb{E} \Big[e^{-\rho T} e^{A^\intercal(T-t)} \hQ  x_T +\int_t^T e^{-\rho s}e^{A^\intercal(s-t) }  (Q x_s + N u_s- \eta \big )ds \Big | \mc{F}_t \Big]  \bigg  \} dt\bigg ].
 \end{multline}
Next, defining the martingale $M$ given in \eqref{eqn:DefM} 
allows us to rewrite \eqref{SLQGderiv2} as stated in \eqref{LQGGatDerivFinal}.
\end{proof}
We next provide the form of the optimal control action in the following Theorem.
\begin{theorem}[LQG Optimal Control Action]\label{thm:LQGoptCntrlRaw} Given \textit{Assumption \ref{ass:singleAgentConvCond}}, the optimal control action for the LQG system \eqref{SLQConvsys}-\eqref{SLQconvCost} is
\begin{equation}\label{CntrlActionRaw}
u^*_t = - R^{-1} \bigg [ N^\intercal x^*_t -n +  B^\intercal  e^{\rho t} \Big (e^{-A^\intercal t}M_t -\int_0^t e^{-\rho s} e^{A^\intercal(s-t) }  (Q x^*_s + N u^*_s -\eta )ds\Big ) \bigg ].
\end{equation}
\end{theorem}
\hfill $\square$
\begin{proof}
From \Cref{thm:EulerIneq} and \Cref{coro:convexAnalysis}, a necessary condition for $u^*\in\mc{U}$ to be the optimal control action is
\begin{align}\label{nessAndSuffCond}
 \langle \mcD{J(u^*)}, \omega  \rangle = 0, \quad  \,\,\, \forall \,\, \omega \in \mc{U}.
 \end{align}
Moreover, as \Cref{ass:singleAgentConvCond} holds,  \eqref{nessAndSuffCond} is also sufficient.

From \eqref{LQGGatDerivFinal}, equation \eqref{nessAndSuffCond} holds if and only if
\begin{align}\label{SingleAgCntrl}
 u^*_t = - R^{-1} \bigg [ N^\intercal x^*_t -n +  B^\intercal  e^{\rho t} \Big (e^{-A^\intercal t}M_t -\int_0^t e^{-\rho s} e^{A^\intercal(s-t) }  (Q x^*_s + N u^*_s -\eta )ds\Big ) \bigg ].
\end{align}

To show this, first by direct substitution of \eqref{SingleAgCntrl} into \eqref{LQGGatDerivFinal} we have $\langle \mcD{J(u^*)}, \omega \rangle=0, \forall \omega\in\mc{U}$. Next, suppose that $\langle \mcD{J(u^*)}, \omega \rangle=0,\, \forall \omega\in\mc{U}$, but \eqref{SingleAgCntrl} does not hold on a measurable set $\mc{B}\subset\Omega\times\mfT$ with strictly positive measure. Then define $\tilde{\omega}$ such that
\begin{multline}\label{omega_Test}
 \tilde\omega_t = e^{-\rho t} N^\intercal x_t^{*} + e^{-\rho t} R u_t^{*} -e^{-\rho t} n \\+ B^\intercal e^{\rho t}\Big(e^{-A^\intercal t} M_t -\int_0^t e^{-\rho s} e^{A^\intercal(s-t)} (Q x_s^{*} + N u_s^{*} - \eta )ds\Big),
\end{multline}
so that by assumption $\mathbb{P}\left(\int_0^T |\tilde{\omega_t}|dt>0\right)>0$.

We next show that $\tilde{\omega}\in\mc{U}$. By definition $\tilde{w}_t$ is $\mc{F}_t$-measurable. Moreover, using Jensen's, Cauchy Schwarz and triangle inequalities, we have
\begin{equation}
\mb{E}\left[\int_0^T \tilde{\omega}_t^\intercal\tilde{\omega}_t dt\right] \leq 2\bigg[\lambda \mb{E}\Vert x_0 \Vert^2 + \kappa \mb{E} \int_0^T \Vert u_t^* \Vert^2 dt +\gamma \bigg] < \infty,
\end{equation}
since $u^*_t \in \mc{U}$, and hence $\mb{E} \int_0^T \Vert u_t^* \Vert^2 dt < \infty$, 
$\mb{E}\Vert x_0 \Vert^2 =\Vert x_0 \Vert^2 < \infty $, 
\begin{align}
\lambda &= \big(2 \alpha \int_0^T \Vert e^{At}\Vert^2 dt + 2\beta \Vert e^{AT}\Vert^2 \big) < \infty,\allowdisplaybreaks
\\
\begin{split}
\kappa &= 4T\Vert N \Vert^2\big(\int_0^T\int_0^T \left\Vert e^{A^\intercal(s-t)}\right\Vert^2dsdt\big)
+ 2\Vert R \Vert^2 
\\&
\qquad+ 2 (\beta+\alpha T) \Vert B \Vert^2 \int_0^T \left\Vert e^{A(T-s)}\right\Vert^2 ds < \infty,
\allowdisplaybreaks
\end{split}
\allowdisplaybreaks\\
\begin{split}
\gamma =&~\Vert n \Vert^2 T + 4T^2 \Vert \eta \Vert^2\int_0^T\int_0^T \left\Vert e^{A^\intercal(s-t)}\right\Vert^2dsdt
\\
&+2(T\alpha + \beta) \left(\left\Vert \int_0^T e^{A(T-s)}b(s)ds \right\Vert^2 + 2T \int_0^T\left\Vert e^{A(T-s)}\sigma(s) \right\Vert^2 ds\right) < \infty,\allowdisplaybreaks
\end{split}
\end{align}
\begin{equation}
\alpha = \Vert N^\intercal\Vert^2+ 4T\Vert Q \Vert^2\left(\int_0^T\int_0^T\left\Vert e^{A^\intercal(s-t)}\right\Vert^2dsdt\right) < \infty,
\end{equation}
and
\begin{equation}
\beta = \Vert B^\intercal \Vert^2 \Vert \hQ \Vert^2 \int_0^T \Vert e^{A^\intercal(T-t)} \Vert^2 dt 
< \infty.
\end{equation}

Finally, upon substituting \eqref{omega_Test} into \eqref{LQGGatDerivFinal}, as $\tilde{\omega}$ does not vanish on $\mc{B}$, we have that
\begin{equation}
\langle \mcD{J(u^*)}, \tilde\omega \rangle = \mathbb{E}\left[\int_0^T \tilde\omega^\intercal_s\tilde\omega_s\,ds \right] > 0,
\end{equation}
which contradicts the assumption that $\langle \mcD{J(u^*)}, \omega \rangle=0$,  $\forall \omega\in\mc{U}$.

To demonstrate the candidate control $u^*$ in \eqref{SingleAgCntrl} is indeed admissible, we first see that as all processes in the rhs are $\mc{F}$-measurable, so is the candidate $u^*$. Moreover, using the triangle inequality and Cauchy-Schwarz inequality we have $\mb{E}[\int_0^T u^{*\intercal}_t u_t^* dt ] < \infty$ and hence $u^* \in \mc{U}$.
\end{proof}
\begin{theorem}[LQG State Feedback Optimal Control]\label{thm:optCntrlAct}  For the LQG system  \eqref{SLQConvsys}-\eqref{SLQconvCost}, the optimal control action is given by the linear state feedback control
\begin{equation}\label{LQGCntrl}
u_t^* = - R^{-1}\left(N^\intercal\,  x_t^*-n+B^\intercal\left[\Pi(t)\, x_t^* + s(t)\right]\right),
\end{equation}
where $\Pi(t)$ and $s(t)$ are deterministic functions satisfying the ODEs
\begin{equation}
\rho \Pi(t) = \dot{\Pi}(t) + \Pi(t) A + A^\intercal \Pi(t) - (\Pi(t)B + N)R^{-1}( B^\intercal \Pi(t) + N^\intercal) + Q  ,\label{LQGRiccatiEq}
\end{equation}
\begin{multline}
 \rho s(t) = \dot{s}(t) + \big [(A- BR^{-1} N^\intercal)^\intercal- \Pi(t) B R^{-1} B^\intercal \big] s(t) \\
 + \Pi(t) (b(t)+BR^{-1}n) +NR^{-1}n-\eta,\label{LQGOffsetEq}
\end{multline}
subject to the terminal conditions $\Pi(T) = \hQ $ and $s(T) = 0$.
\end{theorem}
\hfill$\square$
\begin{proof}
Define the process $(p_t)_{t\in\mfT}$ by
 \begin{align} \label{SLQconvAdjoint}
 p_t = e^{\rho t}\big(e^{-A^\intercal t}M_t - \int_0^t e^{-\rho s}e^{A^\intercal(s-t)} (Q x^*_s + Nu^*_s - \eta)ds \big).
 \end{align}
It corresponds to  the adjoint process for the system \eqref{SLQConvsys}-\eqref{SLQconvCost} resulting from the stochastic maximum principle. Moreover, the martingale representation theorem states there exists an $\mc{F}$-adapted process $(Z_t)_{t\in\mfT}$ such that
 \begin{equation}\label{MrtnglRepSingle}
 M_t = M_0 + \int_0^t Z_s\,dw_s.
\end{equation}
Next, set
 \begin{align} \label{SLQansatz}
 p_t = \Pi(t)\,x^*_t + s(t),
 \end{align}
where $\Pi(t)\in \mbR^{n\times n}$ and $s(t)\in \mbR^n$ are deterministic functions of time which are to be determined. Substituting this expression into \eqref{CntrlActionRaw} gives
\begin{equation}\label{SLQconvControl}
 u^*_t = - R^{-1} \big [ N^\intercal x^*_t -n+  B^\intercal \big(\Pi(t) x^*_t + s(t) \big) \big].
 \end{equation}
Next, applying It\^{o}'s lemma to \eqref{SLQansatz} and \eqref{SLQConvsys}, and using \eqref{SLQconvControl}, $p_t$ satisfies the SDE
 \begin{multline} \label{SLQconvdiff1}
 d{p}_t = \Big[\big (\dot{\Pi}(t) + \Pi(t) A  - \Pi(t) B R^{-1} N - \Pi(t) B R^{-1} B^\intercal \Pi(t) \big) x_t^{*}
 \\
 - \Pi(t) B R^{-1} B^\intercal s(t)
 + \Pi(t) b +\Pi(t) BR^{-1}n +\dot{s}(t)\Big] dt + \Pi(t) \sigma(t) dw_t.
 \end{multline}
 Applying It\^{o}'s lemma to \eqref{SLQconvAdjoint} and \eqref{MrtnglRepSingle} results in
 \begin{align} \label{SLQconvAdjointDiff}
 dp_t = (\rho p_t -A^\intercal p_t - Qx_t^{*} -N u_t^{*} + \eta)dt + e^{\rho t}e^{-A^\intercal t}Z_t dw_t.
 \end{align}
Substituting \eqref{SLQansatz} and \eqref{SLQconvControl} into \eqref{SLQconvAdjointDiff} yields
\begin{multline} \label{SLQconvdiff2}
d{p}_t = \Big[\big(\rho \Pi(t) - Q  + NR^{-1}N^\intercal + N R^{-1} B^\intercal \Pi(t) -A^\intercal \Pi(t)\big)x^*_t + \rho s(t) \\+ (N^\intercal R^{-1} B^\intercal - A^\intercal) s(t) + \eta - N R^{-1} n\Big]dt + q_t dw_t,
\end{multline}
where $q_t = e^{\rho t} e^{-A^\intercal t} Z_t$.

Finally, for \eqref{SLQconvdiff1} and \eqref{SLQconvdiff2} to be equal, the corresponding drift and diffusion terms must be equal. Hence,
\begin{equation}
q_t = \Pi(t) \sigma(t),
\end{equation}
\begin{equation}
\begin{cases}
\rho \Pi(t) = \dot{\Pi}(t) + \Pi(t) A + A^\intercal \Pi(t) - (\Pi(t)B + N)R^{-1}( B^\intercal\Pi(t) + N^\intercal) + Q ,\\
\Pi(T) = \hQ ,
\end{cases}
\end{equation}
and 
 \begin{equation}
\begin{cases}
 \rho s(t) = \dot{s}(t) + \big [(A- BR^{-1} N^\intercal)^\intercal- \Pi(t) B R^{-1} B^\intercal \big] s(t) \\\hspace{4.8cm}+ \Pi(t) \big(b(t)+BR^{-1}n\big)+NR^{-1}n-\eta,\\
s(T) = 0,
\end{cases}
\end{equation}
which completes the proof.
\end{proof}

\begin{remark}[Finite Horizon LQG Systems] Typically, the cost functional for finite horizon LQG systems contain no discount factor, i.e. $\rho=0$, and in this case,  the Riccati and offset equations \eqref{LQGRiccatiEq}-\eqref{LQGOffsetEq} reduce to
\begin{equation}
\left\{
\begin{aligned}
-\dot{\Pi}(t)&= \Pi(t) A + A^\intercal \Pi(t) - (\Pi(t)B + N)R^{-1}( B^\intercal \Pi(t) + N^\intercal) + Q,
\\[0.25em]
-\dot{s}(t) &=  [(A- BR^{-1} N^\intercal)^\intercal- \Pi(t) B R^{-1} B^\intercal ] s(t)\\
&\qquad
+ \Pi(t)(b(t)+ BR^{-1}n)
+N R^{-1}n-\eta,
\end{aligned}
\right.
\end{equation}
subject to the terminal conditions $\Pi(T) = \hQ $, $s(T) = 0$.
\end{remark}
\subsection{Infinite-Horizon LQG Systems}
For infinite horizon LQG systems where the terminal time $T$ in \eqref{SLQconvCost} is set to infinity, the terminal cost is set to zero. In this case,  the cost  functional is
\begin{equation} \label{costInfHorizon}
J(u) =   \tfrac{1}{2}  \mathbb{E} \left[ \int_{0}^{\infty} e^{-\rho t} \Big \{ x_t^\intercal  Qx_t + 2x_t^\intercal N u_t + u_t^\intercal R u_t -2 x_t^\intercal \eta - 2 u_t^\intercal n \Big \}dt  \right],
\end{equation}
The dynamics \eqref{SLQConvsys} remains the same in  infinite horizon LQG systems.
\begin{assumption}\label{ass:detectability}
The pair $(L, A- (\rho /2)\Id_n)$ is detectable where $L = Q^{1/2}$. \end{assumption}
\begin{assumption}\label{ass:stabilizability}
The pair $(A-(\rho/2)\Id_n, B)$ is stabilizable.
\end{assumption}
Given that \textit{Assumptions \ref{ass:detectability}-\ref{ass:stabilizability}} hold, for infinite horizon LQG systems governed by \eqref{SLQConvsys} and \eqref{costInfHorizon}, the optimal control action is given by \eqref{LQGCntrl}, where the steady state Riccati matrix $\Pi$ satisfies an algebraic Riccati equation given by
\begin{equation}
\rho \Pi= \Pi A + A^\intercal \Pi - (\Pi B + N)R^{-1}( B^\intercal \Pi + N^\intercal ) + Q,\\
\end{equation}
and the steady state offset vector $s(t)$ satisfies the differential equation
\begin{equation}
\rho s(t)=\dot{s}(t) +  [(A- BR^{-1} N^\intercal)^\intercal- \Pi B R^{-1} B^\intercal ] s(t)+ \Pi(b(t)+ BR^{-1}n)+N R^{-1}n-\eta,
\end{equation}
see e.g. \cite{Anderson1990}.

\section{Major Minor LQG Mean-Field Game Systems}\label{sec:MMLQGMFG}
In this section, we use the convex analysis method introduced in Section \ref{sec:LQGproblems} to derive the best response strategies for major minor LQG MFG (MM LQG MFG) problems addressed in \cite{Huang2010}. An important distinguishing feature of our approach is that we impose no assumption on the evolution of the mean-field beforehand.

\subsection{Dynamics}

We consider a large population $N$ of minor agents denoted by $\mc{A}_i, i \in \mfN:=\{1,\dots,N\},\,N<\infty$, with a major agent denoted by $\mc{A}_0$, where agents are subject to stochastic linear dynamics and quadratic cost functionals. Each agent is coupled to other agents through their dynamics and cost functional. Both types of coupling occur through  the average state of minor agents, i.e. the empirical mean-field.

Major and minor agents' states are assumed, respectively, to satisfy
\begin{align}
dx_t^0 &= [A_0\,x_t^0 +F_0\, x^{(N)}_t+ B_0\, u_t^0 + b_0(t)]\,dt + \sigma_0 \,dw^0_t, \label{MajorDyn}\\
dx_t^i &= [A_k\, x_t^i +F_k\, x^{(N)}_t + G_k x^0_t + B_k\, u_t^i + b_k(t)]\,dt + \sigma_k \,dw^i_t, \label{MinorDyn}
\end{align}
for $t \in\mfT$, $i \in \mfN$, and the subscript $k$, $k \in \mfK :=\{1,\dots,K\},\, K\leq N$,  denotes the type of a minor agent. Here $x^i_t \in \mathbb{R}^n,~ i \in \mfN_0:=\{0,\dots,N\}$, are the states, $(u^i_t)_{t\in\mfT} \in \mathbb{R}^m,~ i \in \mfN_0$, are the control inputs, $w =\lbrace (w^i_t)_{t\in\mfT}, i \in \mfN_0 \rbrace$ denotes $(N+1)$ independent standard Wiener processes in $\mathbb{R}^r$, where $w^i$ is progressively measurable with respect to
the filtration $\mc{F}^w \coloneqq (\mc{F}_t^{w})_{t\in \mfT}$. Moreover, $x^{(N)}_t := \frac{1}{N} \sum_{i=1}^{N} x^i_t$ denotes the average state of the minor agents. All matrices in \eqref{MajorDyn} and \eqref{MinorDyn} are constant and of appropriate dimension; vectors $b_0(t)$, and $b_k(t)$ are deterministic functions of time.

\begin{assumption} \label{IntialStateAss}
The initial states $\lbrace x^i_0,~ i \in \mfN_0 \rbrace$ defined on $(\Omega, \mathcal{F}, P)$ are identically distributed, mutually independent and also independent of $\mathcal{F}^{w}$, with $\mathbb{E}x^i_0=0$. Moreover, $\sup_{i} \mathbb{E}\Vert x^i_0\Vert^2 \leq c < \infty $, $ i \in \mfN_0$, with $c$ independent of $N$.
\end{assumption} 

\subsubsection{Agents types}
Minor agents are given in $K\le N$ distinct types.
The index set $\mc{I}_k$ is defined as
\begin{equation*}
\mc{I}_k = \lbrace i : \theta_i = \theta^{(k)},~ i \in \mfN \rbrace , \quad k \in \mfK,
\end{equation*}
where $\theta^{(k)}\in\Theta$ for $k\in\mfK$ and $\Theta$ is the parameter set.
The cardinality of $\mc{I}_k$ is denoted by $N_k = |\mc{I}_k|$. Then, $\pi^{N} = (\pi_{1}^{N},...,\pi_{K}^N),~ \pi_k^N = \tfrac{N_k}{N}, ~ k \in \mfK$, denotes the empirical distribution of the parameters $(\theta_1,...,\theta_N)$ obtained by sampling the initial conditions  independently
of the Wiener processes of all agents.
\begin{assumption} \label{ass:EmpiricalDistLimit}
There exists $\pi$ such that $\displaystyle\lim_{N \rightarrow \infty} \pi^N = \pi $ a.s.
\end{assumption}

\subsubsection{Control $\sigma$-Fields}

For any finite $T$, we denote (i) the natural filtration generated by $\mc{A}_i$'s state $(x^i_{t})_{t\in\mfT}$  by $\mc{F}^i\coloneqq (\mc{F}_{t}^i)_{t\in\mfT}$, $ i \in \mfN$, (ii) the natural filtration generated by the major agent's state $(x^0_{t})_{t\in\mfT}$ by $\mc{F}^0:=(\mc{F}^0_t)_{t\in\mfT}$, and (iii) the natural filtration generated by the states of all agents $((x_t^i)_{i \in \mfN_0})_{t\in\mfT}$ by $\mc{F}^g:=(\mc{F}^g_{t})_{t\in\mfT}$.
Next, we introduce three admissible control sets. 

\begin{assumption} [Major Agent Linear Controls]\label{ass:MajorControl}
 For the major agent $\mc{A}_0$, the set of control inputs  $\mc{U}^{0}$ is defined to be the collection of linear state feedback control laws $u^0$ that are adapted to  $\mc{F}^0$ such that $\mb{E}[\int_0^T u_t^{0\intercal}u_t^0\, dt] < \infty$. 
\end{assumption}

\begin{assumption}[Minor Agent Linear Controls] \label{ass: MinorContrAction}
 For each minor agent $\mathcal{A}_i,\, i \in \mfN$, the set of control inputs $\mc{U}^{i}$ is defined to be the collection of linear state feedback control laws adapted to the filtration $\mc{F}^{i,r} \coloneqq (\mc{F}_{t}^{i,r})_{t\in\mfT}$, where $\mc{F}^{i,r} \coloneqq \mc{F}^i \vee \mc{F}^0$, such that  $\mb{E}[\int_0^T u_t^{i\intercal}  u_t^i\, dt] < \infty$.
\end{assumption}
The set of control inputs $\mc{U}^{N}_g$ is defined to be the collection of linear state feedback control laws adapted to $\mc{F}^g$, such that  $\mb{E}[\int_0^T u_t^{i\intercal}  u_t^i\, dt] < \infty, i \in \mfN_0$. 
\subsection{Cost functionals}
We denote the seminorm $\Vert a \Vert_B^2\coloneqq a^\intercal B a$, where $a$ and $B$ are a vector and a matrix of appropriate dimensions, respectively. We also denote $u^{-0} \coloneqq  (u^1,\dots,u^N)$ and $u^{-i} \coloneqq  (u^0,\dots,u^{i-1}, u^{i+1},\dots, u^N)$. Then the major agent's cost functional for the (finite) large population  stochastic game is
\begin{multline} \label{MajorCostLrgPop}
J_{0}^{N}(u^0, u^{-0}) = \tfrac{1}{2} \mathbb{E} \Big[ \Vert x^0_T - \Phi^{(N)}_T\Vert ^2_{\hQ_0}  + \int_{0}^{T}  \Big \lbrace \Vert x^0_t - \Phi^{(N)}_t\Vert ^2_{Q_0} \allowdisplaybreaks\\+ 2{\big(x_t^0-\Phi^{(N)}_t\big)}^\intercal N_0 u_t^0+ \Vert u^0_t \Vert_{R_0}^2 \Big \rbrace dt \Big],
 \end{multline}
where
\begin{equation}
 \Phi^{(N)}_t := H_0\,x^{(N)}_t + \eta_0.
\end{equation}
\begin{assumption} \label{ConvexityCondMajorLQGMFG}
$\hQ_0  \geq 0$, $R_0 >0$, and $Q_0-N_0 R_0^{-1} N_0^\intercal\geq 0$. 
\end{assumption}
A minor agent $\mc{A}_i$, $i\in\mfN$, such that $i\in\mc{I}_k$, cost functional for the  (finite) large population stochastic game is
\begin{multline}
J_{i}^{N}(u^i, u^{-i}) = \tfrac{1}{2}\,
\mathbb{E} \left[
\Vert x^i_T - \Psi^{(N)}_T \Vert_{\hQ_k}^{2}+ \int_{0}^{T} \Big \lbrace \Vert x^i_t - \Psi^{(N)}_t \Vert_{Q_k}^{2} \allowdisplaybreaks
\right.
\\
+ 2{\big(x_t^i- \Psi^{(N)}_t\big)}^{\intercal} N_k u_t^i + \Vert u^i_t \Vert_{R_k}^2 \Big\rbrace \,dt \Bigg],  \label{MinorCostFinitePop}
\end{multline}
where
\begin{equation}\label{PSI}
\Psi^{(N)}_t ~:=~ H_k\, x^0_t + \hH_k  \, x^{(N)}_t + \eta_k.
\end{equation}
\begin{assumption}\label{convexityCondMinorLQGMFG}
$\hQ_k \geq 0$, $R_k >0$, and $Q_k - N_k R_k^{-1} N_k^\intercal\geq 0$, for $k \in \mfK  $.
\end{assumption}
Note that the major and minor agents are coupled to one another through $x^{(N)}_t$ which appears in both the dynamics \eqref{MajorDyn}-\eqref{MinorDyn}  and cost functionals \eqref{MajorCostLrgPop}-\eqref{MinorCostFinitePop}.

\subsection{Solutions to Major Minor LQG MFG Problems}

In the mean-field game methodology with a major agent \cite{NourianSiam2013,Huang2010}, we first solve the infinite population version of the stochastic game. This is achieved by replacing the average terms in the finite population dynamics and cost functional  by their infinite population limit called the mean-field. Then in the case of LQG MFG systems, we extend the major agent's state to include the mean-field and extend the minor agent's state to include the major agent's state and mean-field. The result is a stochastic control problem for each agent (rather than a stochastic game), but the problems are linked through the major agent's state and the mean-field. Next, we solve these stochastic control problems and the resulting fixed point problem to obtain a consistent mean-field. Finally, we apply the infinite population best response strategies to the finite population system and demonstrate that this results in an $\epsilon$-Nash equilibrium \cite{Huang2010}.


By introducing more general cost functionals and state dynamics we generalize Theorem 10
in  \cite{Huang2010}, which provides the control laws that yield the infinite population Nash equilibrium and the resulting finite population $\epsilon$-Nash equilibrium. We present a new proof using a convex analysis method  to derive the best response strategy for every agent without any explicit assumption on the mean-field dynamics -- which contrasts with the classical approaches (see e.g. \cite{Huang2010, KizilkaleTAC2016, FirooziCainesTAC,FJC-arXiv2018,FPC-arXiv2018}).

\begin{assumption} \label{MFEquationSolAss}
The parameters in \eqref{MajorDyn}-\eqref{MinorDyn}, \eqref{MajorCostLrgPop}-\eqref{PSI}, belong to a non-empty set which yields the existence and uniqueness of the solutions ($\Pi_0$, $s_0$, $\Pi_k$, $s_k$, $\bar{A}_k$, $\bar{G}_k$, $\bar{m}_k$) to the resulting set of major-minor mean field fixed-point equations consisting of \eqref{ConsistencyEq1} and \eqref{ConsistencyEq2}.  
 \end{assumption}

\begin{theorem}
[$\epsilon$-Nash Equilibrium for LQG MFG Systems] \label{thm:MMLQGMFG}
Given that \textit{Assumptions \ref{IntialStateAss}-\ref{MFEquationSolAss}} hold, the system of equations \eqref{MajorDyn}-\eqref{PSI} together with the mean-field consistency equations \eqref{ConsistencyEq1}-\eqref{ConsistencyEq2} generate a set of control laws $\mc{U}_{MF}^{\infty}$, with finite sub-families ${\mathcal{U}}_{MF}^{N} ~\triangleq~ \{{u}^{i,*};  i \in \mfN_0\}$, $1 \leq N < \infty$, given by \eqref{SextGMFGLatentSLQconvControl} and \eqref{SextGMFGLatentSLQconvControl-minor}, such that

\begin{enumerate}
\item[(i)]  The set of control laws $\mc{U}_{MF}^{\infty}$ yields the infinite population Nash equilibrium, i.e.,
\begin{equation*}
 J_i^{\infty}(u^{i,*}, u^{-i,*}) = \inf_{u^i \in \mc{U}^{\infty}_{g} }J_i^{\infty} (u^i, u^{-i,*}).
  \end{equation*}
 \item[(ii)] All agent systems $\mc{A}_i, ~i \in \mfN_0$, are second order stable.
\item[(iii)] The set of control laws $\mc{U}_{MF}^{N}$, form an $\epsilon$-Nash equilibrium for all $\epsilon$, i.e., for all $\epsilon>0$, there exists $N(\epsilon)$ such that for all $N \geq N(\epsilon)$ 
\begin{equation*}
J_i^{N}(u^{i,*},  u^{-i,*})-\epsilon \leq\inf_{u^i \in\mc{U}_{g}^{N} } J_i^{N}(u^i, u^{-i,*}) \leq  J_i^{N}( u^{i,*}, u^{-i,*}).
\end{equation*}
\end{enumerate}
 \end{theorem}
\begin{proof}
The proof consists of two parts:
\begin{enumerate}
    \item[(I)] Show that the set of control laws $\mc{U}_{MF}^{\infty}$ forms a Nash equilibrium for the infinite population system.
    \item[(II)] Show that when a finite subset of the control laws $\mc{U}_{MF}^{N}$ is applied to the finite population system, all agent systems are second order stable, and the control laws yield an $\epsilon$-Nash equilibrium.
\end{enumerate}
Part (II) is established by the standard approximation analysis parallel to  that in \cite{Huang2010} and is omitted here. In this section we prove part (I) by using a novel convex analysis approach to retrieve the set of best response strategies $\mc{U}_{MF}^{\infty}$ which yields the Nash equilibrium. This part is further broken into three parts:

\noindent \textit{\textbf{(i)Major Agent: Infinite Population}}
For the infinite population problem for the major agent, we follow the steps below:
\begin{itemize}
\item[(i.a)] Perturb the major agent's control action by $\eps_0$ in the direction $\omega^0\in\mc{U}^0$;
\item[(i.b)] Follow the effect of the major agent's perturbed control action on its own state and every minor agent's state and obtain the resulting perturbed mean-field $\bar{x}^{\eps_0}_t$  (\Cref{fig:majorPert} illustrates the propagation of the perturbation through the system states.);

\begin{figure}
\vspace{-30pt}
\centering
\scalebox{1.2}{
\begin{tikzpicture}[node distance = 2cm]

    \node [block2] (u0) {$u^{0,\epsilon_0}$};
    \node [block2, right of=u0, xshift=0.3cm] (x0) {$x^{0,\epsilon_0}$};
    \node [block2, right of=x0, xshift=0.3cm] (xi) {$x^{i,\epsilon_0}$};
    \node [block2, right of=xi, xshift=1cm] (xbar) {$\bar{x}^{\epsilon_0}$};
    
    \path [line, thick] (u0.0) -- (x0.180);
    
    \path [line, thick] (x0.0) -- (xi.180);
    
    \path [line, thick] (xi.0)([xshift=0.0cm]xi.east) -- (xbar.180)node [xshift=-1cm,above, text width = 1.3cm, text centered] (TextNode) {\tiny Aggregation};

    \path [line, thick] (xbar.290) -| node [yshift = -0.5cm,xshift = 0cm]{} ([xshift=0cm,yshift=-0.5cm]xbar.290) -| (x0.270);
    
    \path [line, thick] (xbar.260) -| node [yshift = -0.3cm,xshift = 0cm]{} ([xshift=0cm,yshift=-0.3cm]xbar.260) -| (xi.270);
 
\end{tikzpicture}}
\caption{Propagation of the perturbation of the major agent's control action through the system states.}\label{fig:majorPert}
\end{figure}
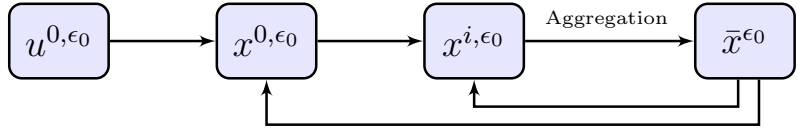

\item[(i.c)] Extend the major agent's state to include the joint dynamics of the  major agent's state and the mean-field;
\item[(i.d)] Use Theorems \ref{thm:LQGGatDeriv}-\ref{thm:optCntrlAct} to obtain the best response strategy for the major agent.
\end{itemize}
We now carry out these steps in turn.

\underline{\bf{Step (i.a)}}: The major agent's state $x^{0,\eps_0}_t$ subject to the perturbed control action $u^0_t+\eps_0\,\omega^0_t$, for $\omega^0_t \in \mc{U}^{0}$, in the infinite population limit satisfies
\begin{equation}
dx^{0,\eps_0}_t = [A_0\, x^{0,\eps_0}_t + F_0^{\pi} \,\bar{x}_t^{\eps_0} + B_0\, (u^0_t + \eps_0 \omega^0_t)+ b_0(t)]dt + \sigma_0\, dw^0_t,
\end{equation}
where $F_0^{\pi}= \pi \otimes F_0 \coloneqq \left[\pi_1 F_0, \dots, \pi_K F_0 \right]$, and $\bar{x}_t^{\eps_0}$ denotes the perturbed mean-field resulting from the major agent's perturbed control action.

\underline{\bf{Step (i.b)}}: The resulting minor agent's state $x^{i, \epsilon_0}$, subject to an arbitrary control $u^i_t\in\mc{U}^{i}$, due to the major agent's perturbation satisfies
\begin{equation} \label{MinorPertbdMajor}
dx^{i, \epsilon_0}_t = [A_k\, x^{i,\epsilon_0}_t dt + F^{\pi}_k\, \bar{x}_t^{\epsilon_0} dt + G_k \,x^{0,\epsilon_0}_t  + B_k \,u^i_t + b_k(t)]dt + \sigma_k\, dw_t^i,
\end{equation}
for all $i\in\mfN$, where $F_k^{\pi}= \pi \otimes F_k \coloneqq \left[\pi_1 F_k, \dots, \pi_K F_k \right]$.
The empirical average of the states of the minor agents of subpopulation $k\in\mfK$ is defined as
\begin{align}
x_t^{(N),k} = \tfrac{1}{N_k} \sum_{i\in\mc{I}_k} x_{t}^{i}.
\end{align}

Subsequently, we define the vector $x^{(N)}_t= [x_t^{(N),1}, x_t^{(N),2},\dots, x_t^{(N),K}]$, where the pointwise in time limit (in quadratic mean) of $x^{(N)}_t$, if it exists, is called the state mean-field of the system and is denoted by $\bar{x}_t^\intercal= [(\bar{x}_t^1)^\intercal, ..., (\bar{x}_t^K)^\intercal]$.

Similarly, the empirical average of the control actions of the minor agents of subpopulation $k\in\mfK$ is defined as
\begin{align}
u_t^{(N),k} = \tfrac{1}{N_k} \sum_{i\in\mc{I}_k} u_{t}^{i}.
\label{eqn:finitepop-meancontrol}
\end{align}

We define the vector $u^{(N)}_t= [u_t^{(N),1}, u_t^{(N),2},\dots, u_t^{(N),K}]$, where the pointwise in time limit (in quadratic mean) of $u^{(N)}_t$, if it exists, is called the control mean-field of the system and is denoted by $\bar{u}_t^\intercal= [(\bar{u}_t^1)^\intercal, ..., (\bar{u}_t^K)^\intercal]$.

We denote the perturbed mean-field for subpopulation $k\in\mfK $ by $(\bar{x}^{k, \eps_0}_t)$. Taking the average of \eqref{MinorPertbdMajor} over $i\in\mc{I}_k$, and then its $L^2$ limit as $N_k\to\infty$, we obtain 
\begin{equation}
 d\bar{x}^{k,\eps_0}_t =
  \left[(A_k\,\mathbf{e}_k+F^{\pi}_k)\,\bar{x}_t^{\eps_0} + G_k \,x^{0,\eps_0}_t + B_k\, \bar{u}_t^{(k)} + b_k(t)
 \right] dt,
\end{equation}
where $\mathbf{e}_k = [0_{n \times n}, ..., 0_{n \times n}, \Id_n, 0_{n \times n}, ..., 0_{n \times n}]$, which has the $n \times n$ identity matrix $\Id_n$ at the $k$th block.
Next, we stack the perturbed subpopulation mean-fields into the perturbed mean-field vector $(\bar{x}_t^{\eps_0})^\intercal := [(\bar{x}^{1,\eps_0}_t)^\intercal, \dots, (\bar{x}^{K,\eps_0}_t)^\intercal]$. Hence $\bar{x}^{\eps_0}$ satisfies the SDE
\begin{equation}\label{MF_perturbed_eps0}
d\bar{x}_t^{\eps_0} = \left(\br{A}\,\bar{x}_t^{\eps_0} + \br{G}\, x^{0,\eps_0}_t + \br{B}\, \bar{u}_t + \br{m}(t)\right)dt,
\end{equation}
where
as described below \eqref{eqn:finitepop-meancontrol}, $\bar{u}_t^\intercal= [(\bar{u}_t^1)^\intercal,\dots, (\bar{u}_t^K)^\intercal]$, and
\begin{equation}\label{MFmatrices}
\br{A} = \begin{bmatrix}
A_1\mathbf{e}_1+ F^{\pi}_1\\
\vdots \\
A_K\mathbf{e}_K+F^{\pi}_K
\end{bmatrix},
\quad
\br{G} = \begin{bmatrix}
G_1 \\
\vdots\\
G_K
\end{bmatrix},
\quad
\br{B} = \begin{bmatrix}
B_1 & &0\\
        & \ddots &\\
       0 & & B_K
\end{bmatrix},
\quad
\br{m}(t) = \begin{bmatrix}
b_1(t)\\
\vdots\\
b_K(t)
\end{bmatrix}.
\end{equation}

\underline{\bf{Step (i.c)}}: We extend the perturbed major agent's state with the perturbed  mean-field to form $(X_t^{0, \eps_0})^\intercal := \left[(x^{0, \eps_0}_t)^\intercal,\,\, (\bar{x}^{\eps_0}_t)^\intercal \right]$, which then satisfies the SDE
\begin{equation}\label{majorExtDynPert}
dX^{0, \epsilon_0}_t = \left(\tilde{A}_0 \,X^{0, \epsilon_0}_t + \mb{B}_0\, u^0_t + \tilde{B}_0\, \bar{u}_t + \epsilon_0\, \mb{B}_0 \,\omega^0_t + \tilde{M}_0 \right) dt + \Sigma_0 \,dW^0_t,
\end{equation}
where
\begin{subequations}
\begin{gather}
 \tilde{A}_0 = \left[ \begin{array}{cc}
A_0 & F_0^{\pi} \\
\br{G} &  \br{A}
\end{array} \right]\!,
\quad \mb{B}_0=\left[ \begin{array}{c} B_0 \\ 0  \end{array}\right]
\!,
\quad
 \tilde{B}_0 = \left[ \begin{array}{c} 0 \\ \br{B}  \end{array}\right]\!, \\
\tilde{M}_0(t) = \left[ \begin{array}{c} b_0(t) \\ \br{m}(t)  \end{array}\right]\!,
\quad
\Sigma_0 = \left[ \begin{array}{cc}
\sigma_0 & 0 \\
0 &  0
\end{array} \right]\!,
\quad W^0_t = \left[ \begin{array}{c}
w^0_t\\
0
\end{array} \right]\!.\label{sysMatMajor}
\end{gather}
\end{subequations}

Given the major agent's perturbed extended state, the corresponding cost functional for the infinite population limit is
\begin{multline}\label{majorExtCost}
J_0^{\infty}(u^0+\eps_0\omega^0) = \tfrac{1}{2}\mb{E} \bigg [  (X_T^{0,\eps_0})^\intercal \mb{G}_0 X_T^{0,\eps_0} + \int_0^T \Big \{(X_s^{0,\eps_0})^\intercal \mb{Q}_0 X_s^{0,\eps_0} \\+ 2(X_s^{0,\eps_0})^\intercal \mb{N}_0 (u_s^0 +\eps_0 \omega^0_s )+ (u^0_s + \eps_0 \omega^0_s)^\intercal R_0 (u^0_s+\eps_0\omega^0_s)\\-2(X_s^{0,\eps_0})^\intercal \bar{\eta}_0 -2 (u^0_s+ \eps_0 \omega^0_s)^\intercal\bar{n}_0 \Big \} ds \bigg],
\end{multline}
where the corresponding weight matrices are
\begin{subequations}
\begin{gather}
\mathbb{G}_0 = \left[\Id_{n}, -H_0^{\pi}\right ]^\intercal \hQ_0   \left [\Id_{n}, -H_0^{\pi}\right],\quad  \mathbb{Q}_0 = \left[\Id_{n}, -H_0^{\pi}\right ]^\intercal Q_0 \left[\Id_{n}, -H_0^{\pi}\right], \\
\mathbb{N}_0 = \left[\Id_{n}, -H_0^{\pi}\right ]^\intercal  N_0, \quad \bar{\eta}_0 = \left[\Id_{n}, -H_0^{\pi}\right ]^\intercal  Q_0 \eta_0, \quad \bar{n}_0 = N_0^\intercal \eta_0,\\
H_0^{\pi} = \left[\pi_1 H_0, \dots, \pi_K H_0\right].
\label{majorExtCostMat}
\end{gather}
\end{subequations}

Following along the lines of the proof for \Cref{thm:LQGGatDeriv} (with $\rho =0$), the G\^ateaux derivative $\mcD{J_0^{\infty}(u^{0})}$ of the major agent's cost functional in the direction of $\omega^0 \in \mc{U}^0$ is given by
 \begin{multline}
\langle \mcD{J_0^{\infty}(u^0)}, \omega^0 \rangle = \mathbb{E} \bigg [\int_0^T (\omega_t^0)^\intercal \bigg\{ \mb{N}_0^\intercal X^{0}_t + R_0 u^0_t - \bar{n}_0\\+ \mb{B}_0^\intercal\Big(e^{-\tilde{A}_0^\intercal t} M_t^0 -\int_0^t e^{\tilde{A}_0^\intercal(s-t) } ( \mb{Q}_0 X^{0}_s + \mb{N}_0 u^0_s -\bar{\eta}_0)ds \Big )  \bigg  \} dt\bigg ],
 \end{multline}
where
\begin{equation}
M_t^0 =  \mathbb{E} \Big[ e^{\tilde{A}_0^\intercal T} \mb{G}_0 X^{0}_ T +\int_0^T e^{\tilde{A}_0^\intercal s } ( \mb{Q}_0 X^{0}_s + \mb{N}_0 u^0_s-\bar{\eta}_0)ds \Big | \mc{F}^0_t \Big],
 \end{equation}
 and the major agent's  unperturbed extended state $X^0_t$ is given by \eqref{majorExtDynPert} where $\eps_0$ is set to zero, i.e.
 \begin{equation}\label{majorExtDynUnPert}
dX^{0}_t = \left(\tilde{A}_0 X^{0}_t + \mb{B}_0 u^0_t + \tilde{B}_0 \bar{u}_t + \tilde{M}_0 \right) dt + \Sigma_0 dW^0_t.
\end{equation}

\underline{\bf{Step (i.d)}}: Since $\mc{D} J_0^{\infty}(u^0)$ has the same structural form as $\mcD{J(u)}$ in \eqref{LQGGatDerivFinal}, from the proof of \Cref{thm:LQGoptCntrlRaw}, the optimal control action for the major agent in the infinite population limit is
\begin{equation}\label{CtrlMajor}
 u^{0,*}_t = - R_0^{-1} \bigg [ \mb{N}_0^\intercal X^{0,*}_t - \bar{n}_0 +  \mb{B}_0^\intercal   \Big ( e^{-\tilde{A}_0^\intercal t}M_t^0 -\int_0^t e^{\tilde{A}_0^\intercal(s-t) }  (\mb{Q}_0 X^{0,*}_s + \mb{N}_0 u^{0,*}_s-\bar{\eta}_0 )ds\Big ) \bigg ].
 \end{equation}
 
Following the steps in the proof of \Cref{thm:optCntrlAct}, we define the major agent's adjoint process $(p^0_t)_{t\in\mfT}$ by
\begin{align}\label{majorAdjoint}
p^0_t = e^{-\tilde{A}_0^\intercal t}M_t^0 -\int_0^t e^{\tilde{A}_0^\intercal(s-t) }  (\mb{Q}_0 X^{0,*}_s + \mb{N}_0 u^{0,*}_s-\bar{\eta}_0 ).
\end{align}
and adopt the ansatz
 \begin{equation}\label{majorAdjointAnsatz}
 p^0_t = \Pi_0(t)\, X^{0,*}_t + s_0(t),
 \end{equation}
where $\Pi_0(t)$ and $s_0(t)$ are deterministic matrices of appropriate dimension. This  provides us with the state feedback control action for the major agent given by
 \begin{equation}\label{SextGMFGLatentSLQconvControl}
 u^{0,*}_t = - R_0^{-1} \big [ \mathbb{N}_0^\intercal X^0_t -\bar{n}_0+  \mathbb{B}_0^\intercal \big(\Pi_0(t) X_t^{0,*} + s_0(t) \big) \big].
 \end{equation}
 
Next applying It\^{o}'s lemma to \eqref{majorAdjoint} and using the martingale representation theorem we find that the major agent's adjoint process satisfies the SDE
\begin{equation}\label{majorAdjointDiff}
d{p}^0_t = \left[-(\mb{Q}_0 + \tilde{A}_0^\intercal\, \Pi_0(t))\,X^{0}_t - \tilde{A}_0^\intercal \,s_0(t) \right]dt+ q^0_t \,dW^0_t,
\end{equation}
where $q_t^0 = e^{-\tilde{A}_0^\intercal t} Z_t^0$, and $(Z_t^0)_{t\in\mfT}$ is an $\mc{F}^0$-adapted process such that $M_t^0 = M^0_0 + \int_0^t Z_s^0\,dW_s^0$. 

Moreover, applying It\^{o}'s lemma to \eqref{majorAdjointAnsatz} and \eqref{majorExtDynUnPert}, $p^0_t$ satisfies the SDE
\begin{multline}
 d{p}_t^0 = \bigg[\big(\dot{\Pi}_0 + \Pi_0 \,\tilde{A}_0 - \Pi_0\, \mathbb{B}_0\, R_0^{-1} \mathbb{B}_0^\intercal\, \Pi_0 \big) \,X^0_t  +\Pi_0 \,\tilde{B}_0\, \bar{u}_t
 \\
 +\Pi_0 \big(\tilde{M}_0(t) - \Pi_0\, \mathbb{B}_0 \,R_0^{-1}\, \mathbb{B}_0^\intercal s_0\big)+ \dot{s}_0(t)\bigg] dt + \Pi_0 \,{\Sigma}_0(t) \,dW^0_t. \label{majorAdjointAnsatzDiff}
\end{multline}

Determining $\Pi_0(t)$ and $s_0(t)$ requires knowing $\bar{u}_t$. To determine this, we next specify the optimal control action for every minor agent.

\noindent \textit{\textbf{(ii)Minor Agent: Infinite Population}}
For the minor agent's infinite population problem we follow the steps below:
\begin{itemize}

\item[(ii.a)] Perturb a minor agent's  control action by $\eps_i$ in the direction $\omega^i\in\mc{U}^i$;

\item[(ii.b)] Follow the effect of the perturbed control action on the major agent's and all minor agents' states and determine the perturbed mean-field $\bar{x}^{\eps_i}_t$ (\Cref{fig:minorPert} illustrates the propagation of the perturbation through the system states.);
\begin{figure}
\vspace{-30pt}
\centering
\scalebox{1.1}{
\begin{tikzpicture}[node distance = 2cm]

    \node [block2] (ui) {$u^{i,\epsilon_i}$};
    \node [block2, right of=ui, xshift=0.3cm] (xi) {$x^{i,\epsilon_i}$};
    \node [block1, right of=xi, xshift=0.78cm] (xbar) {$\bar{x}^{\epsilon_i}=\bar{x}$};
    \node [block1, below of=xi, yshift=0.5cm] (x0) {$x^{0,\epsilon_i}=x^0$};
    
    \path [line, thick] (ui.0) -- (xi.180);
    
    \path [line, thick] (xi.0) -- (xbar.180);
    
    \path [line, thick] (xi.270) -- (x0.90);
 
\end{tikzpicture}}
\caption{Propagation of the perturbation of a minor agent's control action through the system states.}\label{fig:minorPert}
\end{figure}
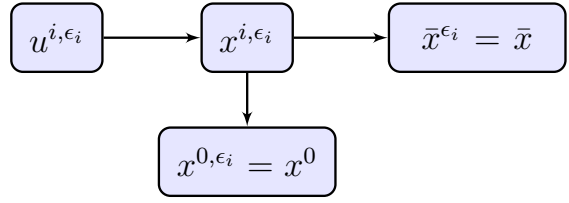

\item[(ii.c)] Extend the minor agents' state to include the joint dynamics of the minor agent, major agent, and  mean-field;
\item[(ii.d)] Use Theorems \ref{thm:LQGGatDeriv}-\ref{thm:optCntrlAct} to obtain the best response strategy for the minor agent.
\end{itemize}
We now carry out these steps in turn.

\underline{\textbf{Step (ii.a)}}: A  minor agent from subpopulation $k$, i.e. $\mc{A}_i,\, i\in \mc{I}_k$,  has a perturbed state $(x^{i,\epsilon_i})_{t\in\mfT}$ subject to the control action $u^i_t+\eps_i \omega^i_t$, where $\omega^i \in \mc{U}^i$, that satisfies the SDE
\begin{equation}\label{minorDynPerturbd}
 dx^{i,\epsilon_i}_t = \left[A_k\,x^{i,\epsilon_i}_t + F^{\pi}_k\, \bar{x}_t^{\epsilon_i} + G_k \,x^{0,\epsilon_i}_t + B_k (u^i_t +\epsilon_i\, \omega^i_t) + b_k(t)\right]dt + \sigma_k\, dw^i_t,
 \end{equation}
where $F^{\pi}_k \coloneqq \pi \otimes F_k$. 

\underline{\textbf{Step (ii.b)}}: To obtain subpopulation $k$'s perturbed mean-field $\bar{x}^{k,\eps_i}_t$, we first take the average of minor agents' states $x_t^{j}$ over subpopulation $k$ where $\mc{A}_j$'s, $i\ne j\in\mc{I}_k$, state satisfies  
\begin{equation}\label{minorDyn_j_Perturbd}
 dx^{j,\epsilon_i}_t = \left[A_k\,x^{j,\epsilon_i}_t + F^{\pi}_k\, \bar{x}_t^{\epsilon_i} + G_k \,x^{0,\epsilon_i}_t + B_k u^j_t + b_k(t)\right]dt + \sigma_k\, dw^j_t,
 \end{equation}
 while $\mc{A}_i$'s state satisfies  \eqref{minorDynPerturbd}, and then we take its limit as $N_k \rightarrow \infty$, to find 
 \begin{equation} \label{MFperturbedMinor}
 d\bar{x}_t^{k,\epsilon_i} = \left[(A_k\, \mathbf{e}_k+ F^{\pi}_k)\, \bar{x}_t^{\epsilon_i} + G_k \,x^{0,\epsilon_i}_t + B_k\, \bar{u}_t^k + b_k(t)\right]dt.
 \end{equation}
 
 Stacking the perturbed subpopulation mean-fields into the vector valued process $\left((\bar{x}_t^{\eps_i})^\intercal := \big[(\bar{x}_t^{1,\eps_i})^\intercal, \dots, (\bar{x}_t^{K, \eps_i})^\intercal \big]\right)_{t\in\mfT}$ and using the above SDE, we find that $\bar{x}_t^{\eps_i}$ satisfies the SDE
 \begin{equation}\label{MFminorPert}
d\bar{x}_t^{\eps_i} =\left( \br{A} \,\bar{x}_t^{\eps_i}  + \br{G} \,x^{0,\epsilon_i}_t + \br{B} \, \bar{u}_t  + \br{m}(t)\right)dt,
\end{equation}
where  the coefficient matrices $\br{A}$, $\br{G}$, $\br{B}$, and $\br{m}(t)$ are given in \eqref{MFmatrices}.

From \eqref{MFperturbedMinor} and \eqref{MFminorPert}, we see that the mean-field, and hence the major agent's dynamics, are not impacted by the perturbation of a minor agent's control action, as heuristically expected, because $\lim_{N_k \rightarrow \infty} \frac{\eps_i \omega^i_t}{N_k} =0$. Therefore $\bar{x}_t^{\eps_i}=\bar{x}_t$ and $x^{0,\eps_i}_t =x^{0}_t$ for any $i \in \mfN$.

\underline{\textbf{Step (ii.c)}}: Next, we extend the  minor agent's state to include the major agent's state and the mean-field, which we assume  to exist. Hence,
$\mc{A}_i$'s perturbed extended state $\left((X^{i, \eps_i}_t)^\intercal ~\coloneqq ~ \left[(x^{i,\eps_i}_t)^\intercal,~(x^0_t)^\intercal,~(\bar{x}_t)^\intercal\right]\right)_{t\in\mfT}$ is  
 \begin{equation}\label{minorExtDynPert}
 dX^{i,\epsilon_i}_t =
 \left(
 \tilde{A}_k\, X^{i,\epsilon_i}_t  + \mb{B}_k\, u^i_t  + \epsilon_i \,\mb{B}_k\, \omega^i_t dt+ \tilde{B}_k\, \bar{u}_t  + \tilde{M}_k(t)
 \right)dt  + \Sigma_k \,dW^i_t,
 \end{equation}
where
\begin{gather}
\tilde{A}_k = \left[ \begin{array}{cc} A_k & [G_k \, \, \, F_k^{\pi}]\\ 0 &\tilde{A}_0- \mathbb{B}_0 R^{-1}_0 \mathbb{N}_0-\mb{B}_0 R_0^{-1} \mb{B}_0^\intercal \Pi_0 \end{array} \right], \quad
 \mb{B}_k = \left[ \begin{array}{c} B_k \\ 0 \end{array}\right], \quad
\tilde{B}_k = \left[ \begin{array}{c} 0 \\ \tilde{B}_0 \end{array}\right],\nonumber\\
 \tilde{M}_k(t) = \left[ \begin{array}{c} b_k(t) \\ \tilde{M}_0(t)-\mb{B}_0 R_0^{-1} \mb{B}_0^\intercal s_0(t) \end{array}\right], \quad
\Sigma_k = \left[ \begin{array}{cc} \sigma_k & 0 \\ 0 & \Sigma_0 \end{array} \right], \quad W^i_t = \left[ \begin{array}{c} w^i_t \\ W^0_t \end{array} \right].\label{sysMatMinor}
\end{gather}

The perturbed infinite population cost functional for $\mc{A}_i,\, i\in\mfN$, is
\begin{multline}\label{minorExtCost}
J_i^{\infty}(u^i+\eps_i \omega^i) = \tfrac{1}{2}\mb{E} \bigg [
(X_T^{i,\eps_i})^\intercal\, \mb{G}_k X_T^{i,\eps_i}
+ \int_0^T \Big \{(X_s^{i,\eps_i})^\intercal \,\mb{Q}_k \,X_s^{i,\eps_i}
\\
+ 2(X_s^{i,\eps_i})^\intercal\, \mb{N}_k\, (u_s^i +\eps_i \omega^i_s)
+ (u_s^i +\eps_i \omega^i_s)^\intercal\, R_k\, (u_s^i +\eps_i \omega^i_s)
\\
-2(X_s^{i,\eps_i})^\intercal \,\bar{\eta}_k
-2 (u_s^i +\eps_i \omega^i_s)^\intercal\,\bar{n}_k
\Big \} ds \bigg],
\end{multline}
where the corresponding weight matrices are
\begin{subequations}
\begin{gather}
\mathbb{G}_k = [\Id_{n}, -H_k, -\hH_k^{\pi}]^\intercal \hQ_k [\Id_{n}, -H_k, -\hH_k^{\pi}],
\\
\mathbb{Q}_k = [\Id_{n}, -H_k, -\hH_k^{\pi}]^\intercal Q_k [\Id_{n},  -H_k, -\hH_k^{\pi}],
\\
\mathbb{N}_k = [\Id_{n}, -H_k, -\hH_k^{\pi}]^\intercal  N_k, \quad \bar{\eta}_k = [\Id_{n}, -H_k, \hH_k^{\pi}]^\intercal  Q_k \eta_k, \quad \bar{n}_k = N_k^\intercal \eta_k,\\
\hH_k^{\pi} = \left[\pi_1 \hH_k, \dots, \pi_K \hH_k \right].
\label{weightMatMinor}
\end{gather}
\end{subequations}%


\underline{\textbf{Step (ii.d)}}: To determine the optimal control $u^{i,*}_t$ for $\mc{A}_i, i\in\mfN$, first, following the lines of the proof for \Cref{thm:LQGGatDeriv}, the G\^ateaux derivative at $u^i\in\mc{U}^i$ in the direction $\omega^i\in\mc{U}^i$ of $\eqref{minorExtCost}$ is
 \begin{multline}
\langle \mcD{J_i^{\infty}(u^i)}, \omega^i \rangle = \mathbb{E} \bigg [\int_0^T (\omega_t^i)^\intercal \bigg\{ \mb{N}_k^\intercal\, X^{i}_t + R_k\, u^i_t -  \bar{n}_k
\\
+ \mb{B}_k^\intercal
\left(e^{-\tilde{A}_k^\intercal t} \,M_t^i
-\int_0^t  e^{\tilde{A}_k^\intercal(s-t) } \left( \mb{Q}_k \,X^{i}_s + \mb{N}_k\, u^i_s  - \bar{\eta}_k\right)ds \right )
\bigg  \} dt\bigg ],
 \end{multline}
where
 \begin{align}
M_t^i =  \mathbb{E} \left[ e^{\tilde{A}_k^\intercal T}\, \mb{G}_k \,X^{i}_ T +\int_0^T e^{\tilde{A}_k^\intercal s } \left( \mb{Q}_k\, X^{i}_s + \mb{N}_k\, u^i_s  - \bar{\eta}_k\right)\,ds  \, \Big | \, \mc{F}^i_t \,\right],
 \end{align}
 and, by setting $\eps_i$ to zero in \eqref{minorExtDynPert}, the unperturbed minor agent's extended state satisfies the SDE
\begin{equation}\label{minorExtDynUnpert}
 dX^{i}_t = \left( \tilde{A}_k\, X^{i}_t  + \mb{B}_k \,u^i_t  + \tilde{B}_k\, \bar{u}_t  + \tilde{M}_k(t) \right) dt + \Sigma_k \,dW^i_t.
 \end{equation}
Again since $\mc{D}J_i^{\infty}(u^i)$ has the same structural form as $\mc{D}J(u)$ in \eqref{LQGGatDerivFinal},
 according to the proof of \Cref{thm:LQGoptCntrlRaw}, the minor agent's optimal control action is
 \begin{equation}
 \begin{multlined}
 \label{cntrlAdjointMinor}
 u^{i,*}_t = - R_k^{-1} \bigg[ \mb{N}_k^\intercal\, X^{i,*}_t
 -\bar{n}_k
 \\
  \left.
 + \mb{B}_k^\intercal \left( e^{-\tilde{A}_k^\intercal t}\,M_t^i
 -\int_0^t e^{\tilde{A}_k^\intercal(s-t) }  \left(\mb{Q}_k \,X^{i,*}_s + \mb{N}_k\, u^{i,*}_s -\bar{\eta}_k\right)\,ds\right ) \right].
 \end{multlined}
 \end{equation}
 Then following the steps in the proof of \Cref{thm:optCntrlAct}, the adjoint process of $\mc{A}_i$ is given by
 \begin{equation}\label{adjointMinor}
p^i_t = e^{-\tilde{A}_k^\intercal t}\,M_t^i -\int_0^t e^{\tilde{A}_k^\intercal(s-t) }  \left(
\mb{Q}_k \,X^{i,*}_s + \mb{N}_k\, u^{i,*}_s -\bar{\eta}_k
\right)\,ds,
\end{equation}
and we adopt the ansatz
\begin{equation}\label{minorAdjointAnsatz}
p^i_t = \Pi_k(t)\, X^{i,*}_t + s_k(t),
\end{equation}
where $\Pi_k(t)$ and $s_k(t)$ are deterministic functions of time. Consequently, the control action \eqref{cntrlAdjointMinor} can be written  in the linear state feedback form
 \begin{align}\label{SextGMFGLatentSLQconvControl-minor}
 u^{i,*}_t = - R_k^{-1} \left[ \mathbb{N}_k^\intercal \, X^{i,*}_t -\bar{n}_k+  \mathbb{B}_k^\intercal \left(\Pi_k(t) \,X_t^{i,*} + s_k(t) \right) \right].
\end{align}

To specify $\Pi_k(t)$ and $s_k(t)$, we apply It\^{o}'s lemma to \eqref{adjointMinor} and use the martingale representation theorem to obtain the SDE
 \begin{equation} \label{minorAdjointDiff}
d{p}^i_t = \left[-\big(\mb{Q}_k + \tilde{A}^\intercal_k \Pi_k(t) \big)X^{i,*}_t - \tilde{A}_k^\intercal s(t) \right]dt + q_t^i dW^i_t,
\end{equation}
where $q_t^i = e^{-\tilde{A}_k^\intercal t} Z_t^i$, and $(Z_t^i)_{t\in\mfT}$ is an $\mc{F}^{i,r}$-adapted process such that $M_t^i = M^i_0 + \int_0^t Z_s^i\,dW_s^i$.

Furthermore, applying It\^{o}'s lemma to \eqref{minorAdjointAnsatz} and \eqref{minorExtDynUnpert} we obtain the SDE
\begin{equation}
 \begin{multlined}
 \label{minorAdjointAnsatzDiff}
 d{p}_t^i = \Big[\left (\dot{\Pi}_k + \Pi_k \,\tilde{A}_k - \Pi_k \,\mathbb{B}_k\, R^{-1}_k \,\mathbb{B}^\intercal_k\, \Pi_k \right) X^i_t- \Pi_k \,\mathbb{B}_k\, R^{-1}_k \mathbb{B}_k^\intercal \,s(t)
 \\
 + \Pi_k\, \tilde{M}_k(t)
 +\Pi_k \,\tilde{B}_k\,\bar{u}_t +\dot{s}_k(t)\Big] dt + \Pi_k\, {\Sigma}_k(t)\, dW^i_t.
 \end{multlined}
 \end{equation}
 
 To determine $\Pi_k(t)$ and $s_k(t)$ we need to first obtain $\bar{u}_t$.

 \noindent \textit{\textbf{(iii) Mean-field and Consistency Equations}}
 To obtain $\bar{u}_t$, we first define
 \begin{gather}
\Pi_k =
\begin{bmatrix}
\Pi_{k,11} & \Pi_{k,12} & \Pi_{k,13} \\
\Pi_{k,21} & \Pi_{k,22} & \Pi_{k,23}\\
\Pi_{k,31} & \Pi_{k,32} & \Pi_{k,33}
\end{bmatrix},\quad 
\mathbb{N}_k = \begin{bmatrix} \mathbb{N}_{k,11} \\ \mathbb{N}_{k,21} \\ \mathbb{N}_{k,31}\end{bmatrix},
\end{gather}
for all $k \in \mfK$, where $\Pi_{k,11}, \Pi_{k,22} \in \mbR^{n \times n}$, $\Pi_{k,33} \in \mbR^{nK \times nK}$, $\mathbb{N}_{k,11}, \mathbb{N}_{k,21} \in \mbR^{n \times m}$, $\mathbb{N}_{k,31} \in \mbR^{nK \times m}$.
Then average \eqref{SextGMFGLatentSLQconvControl-minor} over all $i\in\mc{I}_k$ (i.e., all agents in subpopulation $k$) to obtain
  \begin{align}\label{minorCntrlAve}
 u^{(N_k)}_t = - R_k^{-1} \left( \mb{K}^\intercal \begin{bmatrix}
 x^{(N_k)}_t \\
 x^0_t\\
 \bar{x}_t \end{bmatrix} -\bar{n}_k+\mb{B}_k^\intercal s_k \right).
 \end{align}
 where $\mb{K}:=\mathbb{N}_k + \begin{bmatrix} \Pi_{k,11} & \Pi_{k,12} & \Pi_{k,13} \end{bmatrix}^\intercal B_k $.
 Next \eqref{minorCntrlAve} as $N_k \rightarrow \infty$ converges in quadratic mean to 
 \begin{align}\label{minorCntrlMF}
 \bar{u}^{k}_t = - R_k^{-1} \left(
 \mb{K}^\intercal \begin{bmatrix}
 \bar{x}^{k}_t \\
 x^0_t\\
 \bar{x}_t
 \end{bmatrix} -\bar{n}_k+\mb{B}_k^\intercal s_k \right).
 \end{align}
 
 From \eqref{MF_perturbed_eps0}, the (unperturbed) mean-field equation is given by 
 \begin{equation}\label{MF_uperturbed}
d\bar{x}_t = \left(\br{A}\,\bar{x}_t + \br{G}\, x^{0}_t + \br{B}\, \bar{u}_t + \br{m}(t)\right)dt.
\end{equation}

We substitute \eqref{minorCntrlMF} in the above equation to get 
\begin{equation}\label{MF_uperturbed}
d\bar{x}_t = \left(\bar{A}\,\bar{x}_t + \bar{G}\, x^{0}_t + \bar{m}(t)\right)dt,
\end{equation}
where 
\begin{equation}
\bar{A} = \begin{bmatrix} \bar{A}_1\\ \vdots\\\bar{A}_K\end{bmatrix}, \quad 
\bar{G} = \begin{bmatrix} \bar{G}_1\\ \vdots\\ \bar{G}_K\end{bmatrix}, \quad 
\bar{m} = \begin{bmatrix} \bar{m}_1\\ \vdots \\ \bar{m}_K \end{bmatrix},
\end{equation}
and for $k \in \mfK$
\begin{align}
&\bar{A}_k  = \left[A_k - B_k R_k^{-1} (\mb{N}_{k,11}^\intercal + B_k^\intercal \Pi_{k,11})\right] \mathbf{e}_k\nonumber\\ 
& \qquad \qquad \qquad \qquad \qquad \quad+ F^{\pi}_k - B_k R_k^{-1} ( \mb{N}_{k,31}^\intercal+ B_k^\intercal \Pi_{k,13}),\\
&\bar{G}_k  = G_k -B_k R_k^{-1}(\mb{N}_{k,21}^\intercal + B_k^\intercal \Pi_{k,12}), \\
&\bar{m}_k = b_k+B_kR_k^{-1}\bar{n}_k-B_k R_k^{-1} \mathbb{B}_k^\intercal s_k.
\end{align}

Subsequently, substituting \eqref{minorCntrlMF} into \eqref{majorAdjointAnsatzDiff} and equating \eqref{majorAdjointAnsatzDiff} with \eqref{majorAdjointDiff} results in
 \begin{equation}\label{SextGMFGlatentLQGriccati}
\left\{
\begin{aligned}
-\dot{\Pi}_0
&= \Pi_0 \mb{A}_0 + \mb{A}_0^\intercal \Pi_0 - (\Pi_0 \mb{B}_0 + \mb{N}_0)R_0^{-1}( \mb{B}_0^\intercal \Pi_0 + \mb{N}_0^\intercal ) + \mb{Q}_0,\\
\Pi_0(T) &= \mb{G}_0,
\end{aligned}
\right.
\end{equation}
 \begin{equation} \label{SextGMFGLatentLQGoffset}
\left\{
\begin{aligned}
-\dot{s}_0 =&  \left[\left(\mb{A}_0- \mb{B}_0\,R_0^{-1} \,\mb{N}_0^\intercal\right)^\intercal- \Pi_0\, \mb{B}_0\, R_0^{-1} \,\mb{B}_0^\intercal \right] \,s_0
\\
&\quad+ \Pi_0\left(\mb{M}_0+ \mb{B}_0\,R_0^{-1}\,\bar{n}_0\right)
+\mb{N}_0\, R_0^{-1}\,\bar{n}_0-\bar{\eta}_0,
\\
s_0(T) =& 0,
\end{aligned}
\right.
\end{equation}
where
\begin{gather}
\mb{A}_0 = \begin{bmatrix}
A_0 & F_0^{\pi}\\
\bar{G} & \bar{A}
\end{bmatrix}, 
\quad
\mb{M}_0 = \begin{bmatrix} b_0 \\ \bar{m} \end{bmatrix}.
\end{gather}

Moreover, substituting \eqref{minorCntrlMF} in \eqref{minorAdjointAnsatzDiff}, and equating \eqref{minorAdjointAnsatzDiff} with \eqref{minorAdjointDiff} gives
 \begin{equation}\label{SextGMFGlatentLQGriccati-minor}
 \left\{
\begin{array}{rl}
-\dot{\Pi}_k \!\!\!&= \Pi_k \mb{A}_k + \mb{A}_k^\intercal \Pi_k - (\Pi_k \mb{B}_k + \mb{N}_k)R_k^{-1}( \mb{B}_k^\intercal \Pi_k + \mb{N}_k^\intercal ) + \mb{Q}_k,\\
\Pi_k(T) \!\!\!&= \mb{G}_k,
\end{array}
\right.
\end{equation}
 \begin{equation} \label{SextGMFGLatentLQGoffset-minor}
 \left\{
\begin{array}{rcl}
- \dot{s}_k \!\!\! &=& \!\!\! [(\mb{A}_k- \mb{B}_kR_k^{-1} \mb{N}_k^\intercal)^\intercal- \Pi_k \mb{B}_k R_k^{-1} \mb{B}_k^\intercal ] s_k \\
\!\!\! && \!\!\! \qquad+ \Pi_k( \mathbb{M}_k+\mb{B}_kR_k^{-1}\bar{n}_k)+\mb{N}_kR_k^{-1}\bar{n}_k- \bar{\eta}_k,\\
s_k(T) \!\!\! &=& \!\!\!  0,
\end{array}
\right.
\end{equation}
where
\begin{gather}
\mb{A}_k = \begin{bmatrix} A_k & [G_k \, \, \, F_k^{\pi}]\\ 0 &\mb{A}_0- \mathbb{B}_0 R^{-1}_0 \mathbb{N}_0-\mb{B}_0 R_0^{-1} \mb{B}_0^\intercal \Pi_0 \end{bmatrix}, \quad
 \mb{M}_k = \begin{bmatrix} b_k \\ \mb{M}_0-\mb{B}_0 R_0^{-1} \mb{B}_0^\intercal s_0 \end{bmatrix}.
\end{gather}

Therefore, the set of major minor mean-field fixed-point equations determining $\bar{A}$, $\bar{G}$, $\bar{m}$ (known as the consistency equations) can be written compactly as
\begin{gather}\label{ConsistencyEq1}
\left\{
\begin{array}{rcll}
 -\dot{\Pi}_0 \!\!\! &=& \!\!\!  \Pi_0 \mb{A}_0 + \mb{A}_0^\intercal \Pi_0 - (\Pi_0 \mb{B}_0+ \mb{N}_0)R_0^{-1}( \mb{B}_0^\intercal \Pi_0 + \mb{N}_0^\intercal ) + \mb{Q}_0 ,
\quad 
\Pi_0(T) = \mb{G}_0,
\\
-\dot{\Pi}_k \!\!\! &=& \!\!\!  \Pi_k \mb{A}_k + \mb{A}_k^\intercal \Pi_k - (\Pi_k \mb{B}_k+ \mb{N}_k)R_k^{-1}( \mb{B}_k^\intercal \Pi_k + \mb{N}_k^\intercal ) + \mb{Q}_k,
\quad \!\!
\Pi_k(T) = \mb{G}_k,
\\
\bar{A}_k \!\!\! &=& \!\!\! \left[A_k - B_k R_k^{-1} (\mb{N}_{k,11}^\intercal + B_k^\intercal \Pi_{k,11})\right] \mathbf{e}_k + F^{\pi}_k - B_k R_k^{-1} ( \mb{N}_{k,31}^\intercal+ B_k^\intercal \Pi_{k,13})\,,
 \\
\bar{G}_k  \!\!\! &=& \!\!\!G_k -B_k R_k^{-1}(\mb{N}_{k,21}^\intercal + B_k^\intercal \Pi_{k,12})\,,
\end{array}
\right.
\end{gather}
\begin{gather}\label{ConsistencyEq2}
\left\{
\begin{array}{rclr}
-\dot{s}_0 \!\!\! &=& \!\!\!  [(\mb{A}_0- \mb{B}_0R_0^{-1} \mb{N}_0^\intercal)^\intercal- \Pi_0 \mb{B}_0 R_0^{-1} \mb{B}_0^\intercal ] s_0
\\
& & \quad \quad\quad + \Pi_0(\mb{M}_0+ \mb{B}_0R_0^{-1}\bar{n}_0)+\mb{N}_0 R_0^{-1}\bar{n}_0-\bar{\eta}_0,
&
s_0(T) = 0,
\\
-\dot{s}_k \!\!\! &=& \!\!\!  [(\mb{A}_k- \mb{B}_kR_k^{-1} \mb{N}_k^\intercal)^\intercal- \Pi_k \mb{B}_k R_k^{-1} \mb{B}_k^\intercal ] s_k \\
& & \quad \quad \quad
+ \Pi_k( \mathbb{M}_k+\mb{B}_kR_k^{-1}\bar{n}_k)+\mb{N}_kR_k^{-1}\bar{n}_k- \bar{\eta}_k,
& s_k(T) = 0,
 \\
\bar{m}_k \!\!\! &=& \!\!\!  b_k+B_kR_k^{-1}\bar{n}_k-B_k R_k^{-1} \mathbb{B}_k^\intercal s_k,
\end{array}
\right.
\end{gather}
for all $k\in\mfK$.
\end{proof}

\subsection{Infinite-Horizon LQG MFG Systems}
For infinite horizon LQG MFG systems where the terminal time is set equal to infinity, and hence the terminal cost turns to zero, the major agent's infinite horizon cost functionals is
\begin{equation} \label{MajorCostLrgPopInfHorizon}
J_{0}^{N}(u^0, u^{-0}) = \tfrac{1}{2} \mathbb{E} \Big[ \int_{0}^{\infty}
\!\!\! e^{- \rho t} \Big \lbrace \Vert x^0_t - \Phi(x^{(N)}_t)\Vert ^2_{Q_0}
+ 2{\big(x_t^0-\Phi(x^{(N)}_t)\big)}^{\intercal} N_0 u_t^0+ \Vert u^0_t \Vert_{R_0}^2 \Big \rbrace dt \Big].
\end{equation}
Similarly, the discounted infinite horizon cost functional for  $\mc{A}_i,\,i\in \mfN$, is given by
\begin{equation}
J_{i}^{N}(u^i, u^{-i}) = \tfrac{1}{2}\mathbb{E} \Big[ \int_{0}^{\infty}\!\!\! e^{-\rho t} \Big \lbrace \Vert x^i_t - \Psi(x^{(N)}_t) \Vert_{Q_k}^{2}
+ 2{\big(x_t^i- \Psi(x^{(N)}_t)\big)}^{\intercal} N_k u_t^i + \Vert u^i_t \Vert_{R_k}^2 \Big\rbrace dt \Big].  \label{MinorCostFinitePopInfHorizon}
\end{equation}
The major-minor agent dynamics in \eqref{MajorDyn}-\eqref{MinorDyn} remain the same in the infinite horizon LQG MFG systems.

For the results of \Cref{thm:MMLQGMFG} to hold for the infinite-horizon LQG MFG systems, we add the following to \textit{Assumptions \ref{IntialStateAss}-\ref{convexityCondMinorLQGMFG}}. 

\begin{assumption} \label{MFEquationSolAss-infHorizon}
The parameters in \eqref{MajorDyn}-\eqref{MinorDyn} and \eqref{MajorCostLrgPopInfHorizon}-\eqref{MinorCostFinitePopInfHorizon} belong to a non-empty set which yields the existence and uniqueness of the solutions ($\Pi_0$, $s_0$, $\Pi_k$, $s_k$, $\bar{A}_k$, $\bar{G}_k$, $\bar{m}_k$) to the resulting set of major-minor mean field fixed-point equations, and for which the matrices 
\begin{align}
&\mb{A}_0 - \mb{B}_0 R_0^{-1}\mb{B}_0^\intercal \Pi_0 - \tfrac{\rho}{2}\Id_{n+nK},\\
&\mb{A}_k - \mb{B}_k R_k^{-1}\mb{B}_k^\intercal \Pi_k - \tfrac{\rho}{2}\Id_{2n+nK}, \quad k \in \mfK,
\end{align}
are asymptotically stable, and 
\begin{equation}
\sup_{t\geq 0, k \in \mfK} e^{-\tfrac{\rho}{2}t}\left( |s_0(t) | + |s_k(t)| + |\bar{m}_k(t)| \right) < \infty.
\end{equation}
 \end{assumption}

\begin{assumption}\label{ass:detectability_MF}
The pair $(L_{0}, \mathbb{A}_0 - (\rho /2)\Id_{n+nK})$ is detectable, and for each $k \in \mfK  $, the pair $(L_{k}, \mathbb{A}_k - (\rho /2)\Id_{2n+nK})$ is detectable, where $L_0 = Q_0^{1/2}[\Id_n, -H_0^\pi]$ and $L_k = Q^{1/2}_k[\Id_n, -H_k, -\hH_k^\pi]$. \end{assumption}
\begin{assumption}\label{ass:stabilizability_MF}
The pair $(\mathbb{A}_0-(\rho/2)\Id_{n+nK}, \mathbb{B}_0)$ is stabilizable and $(\mathbb{A}_k-(\rho/2)\Id_{2n+nK}, \mathbb{B}_k)$ is stabilizable for each $k\in \mfK  $.
\end{assumption}
Then for the major agent's system \eqref{MajorDyn}, \eqref{MajorCostLrgPopInfHorizon}, the best response strategy is given by \eqref{SextGMFGLatentSLQconvControl}, where the Riccati matrix $\Pi_0$ satisfies an algebraic Riccati equation given by
\begin{equation}\label{majorRicInfhorizon}
\rho \Pi_0= \Pi_0 \mb{A}_0 + \mb{A}_0^\intercal \Pi_0 - (\Pi_0 \mb{B}_0 + \mb{N}_0)R_0^{-1}( \mb{B}_0^\intercal \Pi_0 + \mb{N}_0^\intercal) + \mb{Q}_0,\\
\end{equation}
and the offset vector $s_0$ satisfies the differential equation
\begin{multline}
\rho s_0(t)=\dot{s}_0(t) +  [(\mb{A}_0- \mb{B}_0R_0^{-1} \mb{N}_0^\intercal)^\intercal- \Pi_0 \mb{B}_0 R_0^{-1} \mb{B}_0^\intercal ] s_0(t)\\+ \Pi_0(\mb{M}_0(t)+ \mb{B}_0R_0^{-1}\bar{n}_0)+\mb{N}_0 R_0^{-1}\bar{n}_0-\bar{\eta}_0.\end{multline}

 Similarly, for $\mc{A}_i$'s system  \eqref{MinorDyn}, \eqref{MinorCostFinitePopInfHorizon}, $i \in \mfN$, the best response strategy is given by \eqref{SextGMFGLatentSLQconvControl-minor}, where the Riccati matrix $\Pi_k$ and the offset matrix $s_k$ satisfy the following equations.
\begin{equation}\label{minorODEInfHorizon}
\left\{
\begin{array}{rclr}
\rho {\Pi}_k
\!\!\!&=&\!\!\! \Pi_k \mb{A}_k + \mb{A}_k^\intercal \Pi_k - (\Pi_k \mb{B}_k+ \mb{N}_k)R_k^{-1}( \mb{B}_k^\intercal \Pi_k + \mb{N}_k^\intercal ) + \mb{Q}_k,
\\
\rho s_k(t)
\!\!\!&=&\!\!\!
\dot{s}_k(t) +  [(\mb{A}_k- \mb{B}_kR_k^{-1} \mb{N}_k^\intercal)^\intercal- \Pi_k \mb{B}_k R_k^{-1} \mb{B}_k^\intercal ] s_k(t)
\\
& & \qquad \qquad + \Pi_k( \mathbb{M}_k(t)+\mb{B}_kR_k^{-1}\bar{n}_k)+\mb{N}_kR_k^{-1}\bar{n}_k- \bar{\eta}_k,

\end{array}
\right.
\end{equation}
for all $k\in\mfK$.

Following the lines of the proof of \Cref{thm:MMLQGMFG}, subject to \textit{Assumptions \ref{IntialStateAss}-\ref{ass:stabilizability_MF}} it can be shown that the set of strategies $\mc{U}_{MF}^{N} \triangleq \{ {u}^{i,*};  i \in \mfN_0 \}$, $1\leq N < \infty$, given by \eqref{SextGMFGLatentSLQconvControl}, \eqref{SextGMFGLatentSLQconvControl-minor}, and \eqref{majorRicInfhorizon}-\eqref{minorODEInfHorizon}, yields an $\epsilon$-Nash equilibrium for the infinite-horizon LQG MM MFG Systems.



\section*{Acknowledgment}
The authors would like to thank Minyi Huang for his productive comments concerning this work.
 \bibliographystyle{elsarticle-num} 
 \bibliography{bib_Dena_10June20}


\end{document}